\documentclass[journal, 12pt,onecolumn,draftclsnofoot]{IEEEtran}
\usepackage{amsmath,amssymb,mathtools,amsthm,amsfonts,array}
\usepackage{subfigure,graphicx,graphics,color,psfrag}
\usepackage{cite,balance,caption}
\allowdisplaybreaks

\usepackage{algorithm,algorithmicx, algpseudocode}
\usepackage{accents,bm,url}
\usepackage[english]{babel}
\usepackage{multirow,enumerate,cases}
\usepackage{stfloats,dsfont,soul,fancyhdr,hhline}

\usepackage[top=0.75in, bottom=0.75in, left=0.7in, right=0.7in]{geometry}

\newtheorem{theorem}{Theorem}
\newtheorem{lemma}{Lemma}
\newtheorem{remark}{Remark}

\begin{document}

\def\lc{\left\lceil}   
\def\rc{\right\rceil}
\def\lf{\left\lfloor}   
\def\rf{\right\rfloor}
\include{header}
\setlength{\topskip}{-3pt}
\addtolength{\belowcaptionskip}{-2mm}
\setlength{\abovecaptionskip}{-2pt}

\title{Analysis of IRS-Assisted Downlink Wireless Networks over Generalized Fading}

\author{Yunli LI, and 
        Young Jin CHUN \IEEEmembership{Member, IEEE}
\thanks{This work was supported in part by the Early Career Scheme under Project 9048208 established under the University Grant Committee (UGC) of the Hong Kong Special Administrative Region (HKSAR), China; in part by the City University of Hong Kong (CityU), Startup Grant 7200618; in part by the CityU, Strategic Research Grant 7005467; and in part by the RMGS Donation Grant 9229080. (Corresponding author: Y. J. Chun).}%
\thanks{Y.~L.~Li and Y.~J.~Chun are with the Department of Electrical Engineering, City University of Hong Kong, Hong Kong, China. Y. J. Chun is also with the Center for Internet of Things, City University of Hong Kong Dongguan Research Institute, Dongguan 523000, China. \\(e-mail: yunlili2-c@my.cityu.edu.hk; yjchun@cityu.edu.hk)}}

\markboth{}%
{Shell \MakeLowercase{\textit{et al.}}: Bare Demo of IEEEtran.cls for IEEE Journals}

\maketitle

\begin{abstract}
Future wireless networks are expected to provide high spectral efficiency, low hardware cost, and scalable connectivity. An appealing option to meet these requirements is the intelligent reflective surface (IRS), which guarantees a smart propagation environment by adjusting the phase shift and direction of received signals. However, the composite channel of IRS-assisted wireless networks, which is composed of a direct link and cascaded link aided by the IRS, has made it challenging to carry out system design and analysis. This motivates us to find tractable and accurate channel modeling methods to model multiple types of channels. To this end, we adopt mixture Gamma distributions to model the direct link, the cascaded link, and the mixture channel. Moreover, this channel modeling method can be applied to various transmission environments with an arbitrary type of fading as the underlying fading of each link. Additionally, a unified stochastic geometric framework is introduced based on this tractable channel model. First, we derived distributions of the cascaded link and the mixture channel by proving multipliability and quadratic form of mixture Gamma distributed channels. Then, we carried out a stochastic geometric analysis of the system performance of the IRS-assisted wireless network with the proposed channel modeling method. Our simulation shows that the mixture Gamma distributed approximation method guarantees high accuracy and promotes the feasibility of system performance analysis of IRS-assisted networks with complicated propagation environments, especially with a generalized fading model. Furthermore, the proposed analytical framework provides positive insights into the system design regarding reliability and efficiency.
\end{abstract}


%
\IEEEpeerreviewmaketitle

\section{Introduction}

For 6G wireless communication, such as Terahertz (THz) systems, transformative solutions to a fully connected world are expected to drive the surge for accommodating the complicated propagation environment, boosting spectral efficiency and providing high reliability. When the 6G system mitigating to higher frequency, these requirements are huge challenges due to fast attenuation and weak penetration \cite{boulogeorgos2018terahertz}. One promising approach that emerged recently is the notion of an intelligent communication environment (ICE). ICE is able to control the propagation environment to adapt to the complicated propagation environment, enhance the reliability, and enlarge the coverage cost-effectively \cite{akyildiz20206g}. 

Various technologies have been proposed to achieve ICE, and one popular and practical solution is the intelligent reflective surface \cite{wu2019towards}, which is also known as reconfigurable intelligent surface (RIS) \cite{di2020smart}, and large-scale intelligent surface (LIS) \cite{han2019large}. The IRS consists of a massive number of passive reflective elements on its planar surface and a control part that adjusts each element's phase shift and direction. In contrast to traditional RF chains, the passive IRS elements only reflect signals without additional active processing, which facilitates the IRS to be deployed easily and cost-efficiently. It is worth noting that the passive-IRS potentially achieves a quantum leap improvement for self-interference and noise amplification compared to active relays and surfaces. In other words, IRS is a revolutionary technology that can achieve high spectrum and energy efficiency communications with low costs \cite{wu2021intelligent}. Based on these advantages, we will adopt passive-IRS in the sequel.

\subsection{Related works}

Spurred by the massive popularity of IRS, considerable researches have been undertaken in the latest decades regarding each aspects of IRS. There are relatively sufficient works about the link-level analysis of IRS-assisted wireless communication systems \cite{wu2021intelligent}. In \cite{jia2020analysis}, the direct link from Base Station (BS) to User Equipment (UE) was modeled as Rayleigh fading while links aided by IRS were modeled as Rician fading, and the IRS worked with quasi-static phase shift design. In \cite{atapattu2020reconfigurable}, the authors analyzed the network performance of IRS-assisted two-way communications between two users over Rayleigh fading by approximating the double Rayleigh fading with a Gamma distribution through moment matching. In contrast, network-level research is still scarce. In \cite{lyu2021hybrid}, the network-level performance of IRS-assisted downlink network was analyzed over Rayleigh fading by approximating the cascaded channel as Complex Normal (CN) distribution through Central Limit Theorem (CLT). Additionally, Gamma distribution is introduced to approximate the received signal power. However, the existing works are focused on simple fading models, and the channel models on the cascaded link and the mixture channel are scarce.



\subsection{Motivation}

The previous system performance analysis mainly worked on Rayleigh fading due to its simplicity and tractability \cite{lyu2021hybrid}. Nonetheless, given the diverse range of operating environments of 6G, they may also be subject to clustering of scattered multipath contributions, \textit{i.e.}, propagation characteristics which are quite dissimilar to conventional Rayleigh fading environments \cite{chun2017stochastic}. Aside from small-scale fading, large-scale fading and random shadowing caused by obstacles in the local environment or human body movements can impact link performance via fluctuating the received signals, which can not be ignored in future wireless communications systems, \textit{i.e.}, mm-Wave wireless communications and THz wireless communications \cite{chun2017comprehensive}. As such, it is essential to extend the analysis of IRS-assisted wireless communication systems to generalized fading channels with novel channel modeling methods.

Moreover, the mixture channel between typical UE and its serving BS consists of two types of link: direct link (BS$\rightarrow$UE), and cascaded link (BS$\rightarrow$IRS$\rightarrow$UE). Statistical characterization of the cascaded and mixture channel in IRS-assisted networks involves highly specialized functions, such as Fox-H or Meijer G-function, even with the simplest Rayleigh fading on each individual links, which causes the performance analysis of IRS-assisted wireless network to be challenging. Considerable researches have been conducted to analyze over asymmetric cascaded channels in relay-assisted networks: mixed Rayleigh and Rician \cite{duong2007effect}, mixed Nakagami-$m$ and Rician \cite{gurung2010performance}, mixed $\eta - \mu$ and $\kappa - \mu$ fading channels \cite{6473919}. Furthermore, there are some approximation works on a symmetric cascaded fading channel in MIMO communications: N*Nakagami-$m$ distribution for Nakagami-$m$ fading channels \cite{karagiannidis2007n}. In addition, \cite{di2009comprehensive} analyzed the dual-hop link over generalized fading channels by leveraging properties of Meijer-G function. While significant advances have been made by previous researches, most of the existing literature approximated the cascaded channels by CN distribution based on CLT or modeled the channels with Meijer-G function. Besides, the system performance analysis is mainly based on the ratio of signal power and noise power (SNR), and ignored the interference, which is an essential part in future dense networks. Although \cite{lyu2021hybrid} has considered the interference effect, the channel model adopted is still approximated by CN distributions through CLT, with Rayleigh as the underlying fading model. Therefore, an approximation model for cascaded link and mixture channel with high accuracy for generalized fading models, is critical for evaluating IRS-assisted network system performance metrics of interest, especially for B5G and 6G.


\subsection{Contributions}

Motivated by the above, we emphasize addressing the modeling of cascaded link, mixture channel, and system-level performance analysis for IRS-assisted wireless networks in this work. We extend the research from Rayleigh fading to arbitrary underlying fading types, such as Nakagami-$m$, Rician, $\kappa$ - $\mu$, and $\kappa$ - $\mu$ shadowed fading, which is a generalized channel modeling method fitting to various networks. We also evaluated the performance metrics with a uniform stochastic geometric framework. The main contributions of this work are summarized as below:

\begin{enumerate}
	\item First and foremost, modeling the channel gain tractably for the cascaded link and mixture channel with high accuracy is essential for the analysis of IRS-assisted networks. In this work, we introduced a general channel modeling method for multiple types of channels in IRS-assisted networks utilizing the multipliability and quadratic form of the mixture Gamma distribution. Thus, we approximated the direct channel, cascaded channel and mixture channel by mixture Gamma distributions with accuracy less than $10^{-5}$. This mixture Gamma channel modeling method works for arbitrary underlying fading and includes single channel, double channel, and mixture channel as a special case.
	
	\item Then, we derived the distribution of conditional received signal power, and Laplace transform of the aggregated interference using stochastic geometry under three operation modes: a) one IRS is associated with the typical UE and performs beamforming whilst other related IRSs randomly scattering the received signals; b) all related IRSs randomly scatter signals to the typical UE without beamforming; c) there is no related IRS, and the whole network works as a traditional network. 
	
	\item Next, we introduced a unified analytical framework for the IRS-assisted network performance evaluation based on the proposed mixture Gamma channel modeling method, where interested performance metrics can be expressed as functions of the ratio of signal power and interference power plus noise power (SINR). Furthermore, we  illustrated several performance metrics, such as spectrum efficiency, SINR moments, and outage probability by invoking their corresponding SINR functions. 
	
	\item Finally, we verified our channel model by Monte-Carlo simulation, which illustrated that the proposed channel modeling method fits well for multiple types of channel with high accuracy. As such, the proposed modeling method can be applied to various wireless systems. Our analysis provides insights on system design and further optimization of the IRS-assisted networks.
	
\end{enumerate}

\subsection{Organizations}

The remaining paper is organized as below. In section \ref{sectionSystemM}, we introduced the system model, association policy, and channel models. In section \ref{sectionMG}, we evaluated channel modeling method of the single link, cascaded link, and mixture channel by proving the multipliability and quadratic form of mixture Gamma distributed channels. In section \ref{section:Perf}, we derived the channel power gain and Laplace transforms of the aggregated interference power under three operation modes and introduced a unified stochastic geometric system performance analysis framework for the IRS-assisted network. In section \ref{sectionNumerical}, we provided simulations to verify our theoretical analysis. In section \ref{sectionConclusion}, we concluded the whole work.

\section{SYSTEM MODEL} \label{sectionSystemM}

We consider an IRS-assisted multi-cell wireless network, where the IRSs are deployed to assist the downlink transmission as shown in Fig.~\ref{fig:IRS-assisted muti-cell wireless network (DL)}. The locations of BSs are modeled by an independent two dimensional (2D) homogeneous Poisson Point Process (HPPP), denoted as $\Lambda_{\rm B}$ with node density $\lambda_{\rm B}$. The locations of IRSs and UEs are modeled as independent 2D-HPPPs, denoted as  $\Lambda_{\rm I}$ with density $\lambda_{\rm I}$ and $\Lambda_{\rm U}$ with density $\lambda_{\rm U}$, respectively. Without loss of generality, we assume that a typical UE, denoted by ${\rm UE}_{0}$, is located at the origin and each BS has an infinitely backlogged queue. The channel is assumed to be frequency-flat and constant while the channel may vary over different frequency bands or time slots \cite{lyu2021hybrid}. To facilitate the analysis, we employ orthogonal multiple access, implying no intra-cell interference. We summarized the common notations used in this paper  in Table~\ref{Tab:notations}.

\subsection{BS and IRS association policy }

We adopt a general association model for BS where each UE connects to the BS that provides strongest long term received signal power without small-scale fading, denoted as ${\rm BS}_{0}$, which is equivalent to connecting to the nearest BS. As such, PDF of the distance between ${\rm BS}_{0}$ and ${\rm UE}_{0}$, denoted as $d_{\rm BU}$, could be derived from the void probability of a 2D HPPP. The PDF of $d_{\rm BU}$ is given by 
\begin{equation} 
	f_{d_{\rm BU}}(d)=2\pi \lambda_{\rm B} d e^{-\lambda _{\rm B}\pi  d^{2}} \label{con:PDFofDistanceBSUE}.
\end{equation}

For the IRS association policy, we assume that at most one IRS is associated between ${\rm UE}_{0}$ and ${\rm BS}_{0}$. As \cite{you2022deploy} shows the optimal deployment location for a single associated IRS is in the vicinity of either ${\rm UE}_{0}$ or ${\rm BS}_{0}$. However, the communications suffer severe product path loss when the link distances between nodes are too large. For this reason, we define a service area of each IRS, which is a circle with radius $D_{1}$. Further, we define an interference area, within which the not associated UEs can receive interference signals from the IRS. The radius of this interference area is denoted as $D_{2}$ \cite{lyu2021hybrid}. Since the deployment of a large-scale centralized IRS is not practical, the association policy adopted in this work is connecting the ${\rm UE}_{0}$ with its nearest IRS located within service area, denoted as ${\rm IRS}_{0}$. Based on the distance to ${\rm UE}_{0}$, the $\Lambda_{\rm I}$ is thinned into three small point processes: the serving IRSs (denoted as $\Lambda_{\rm I,S} \triangleq \{ {\rm IRS}_{0} \}$), the interfering IRSs (denoted as $\Lambda_{\rm I,F}$), the noise IRSs (denoted as $\Lambda_{\rm I,N}$). As such, this IRS association policy contains three operation modes in terms of the distance between ${\rm UE}_{0}$ and its nearest IRS:
\begin{itemize}
	\item \textbf{Mode 1.} If the distance between ${\rm UE}_{0}$ and its nearest IRS is less than $D_{1}$, ${\rm UE}_{0}$ associates to its nearest IRS. 
	\item \textbf{Mode 2.} If the distance between ${\rm UE}_{0}$ and its nearest IRS is larger than $D_{1}$ and less than $D_{2}$, ${\rm UE}_{0}$ does not connect with any IRS. The IRSs, whose distance to ${\rm UE}_{0}$ is less than $D_{2}$, randomly scatter any received signals, which contribute to the interference.
	\item \textbf{Mode 3.} If the distance between ${\rm UE}_{0}$ and its nearest IRS is larger than $D_{2}$ , the random scattering can be ignored or treated as an additive Gaussian white noise (AWGN).
\end{itemize}
As illustrated in Fig.~\ref{fig:IRS-assisted muti-cell wireless network (DL)}, for Mode 1, if there is one IRS associated with ${\rm UE}_{0}$, there are two types of links between the serving BS (${\rm BS}_{0}$) and ${\rm UE}_{0}$, including ${\rm BS}_{0} \rightarrow {\rm UE}_{0}$ (the direct link) and ${\rm BS}_{0} \rightarrow {\rm IRS}_{0} \rightarrow {\rm UE}_{0}$ (the cascaded link). According to void probability of the 2D HPPP, the PDF of the inter-node distance across ${\rm IRS}_{0}$ and the ${\rm UE}_{0}$ (denoted as $d_{\rm IU}$) is given by 
\begin{equation} 
	f_{d_{\rm IU}}({d})=2\pi \lambda _{\rm I} {d} e^{-\lambda_{\rm I}\pi {d}^{2}} \label{con:PDFofDistanceIRSUE}.
\end{equation}

To ensure tractability of the analysis, we assume that the distance and its distribution from ${\rm BS}_{0}$ to ${\rm UE}_{0}$ are identical with that from ${\rm BS}_{0}$ to ${\rm IRS}_{j}$, $d_{\rm BI}^{(j)} \approx d_{\rm BU} $, where $d_{\rm BI}^{(j)}$ denotes the distance between ${\rm BS}_{0}$ and ${\rm IRS}_{j}$, which is widely adopted in literature \cite{lyu2021hybrid}. Additionally, in Appendix~\ref{appendix:Distance}, we provide math proof for the first time, by deriving the conditional PDF, CDF, and mean of the distance between ${\rm IRS}_{j}$ and ${\rm BS}_{0}$, denoted as $d_{\rm BI}^{(j)}$. Besides, there is no need to derive the unconditional PDF and CDF of $d_{\rm BI}^{(j)}$, since only the corresponding $d_{\rm BI}^{(0)}$ for each $d_{\rm IU}^{(0)}$ and $d_{\rm BU}$ pair is meaningful. In other words, the absolute position of IRS is meaningless while the relative position for a given UE matters. This result provides positive insights for operators on the deployment of IRS.

\subsection{Channel model}

For simplicity, assume that both the BSs and UEs are equipped with a single antenna while each IRS consists of $N$ reflective elements. Let $h_{\rm BU} = \sqrt{\zeta_{\rm BU}} g_{\rm BU}$ denotes the channel from BS to UE, where $\zeta_{\rm BU}\triangleq \epsilon  d_{\rm BU}^{-\alpha_{\rm BU}}$ denotes the BS-UE path-loss with $\epsilon$ representing the reference channel power gain at a distance of $1$ m,  $d_{\rm BU}$ being the BS-UE distance, and $\alpha_{\rm BU}$ being the corresponding path-loss exponent. Moreover, $g_{\rm BU}$ denotes the small-scale fading channel. Similarly, the BS$\to$IRS and IRS$\to$UE channels, denoted by ${\bf h}_{{\rm BI}}\in\mathbb{C}^{N\times 1}$ and ${\bf h}^H_{{\rm IU}}\in\mathbb{C}^{1\times N}$, respectively, can be modeled as
\begin{align} 
	{\bf h}_{\rm BI}=\sqrt{\zeta_{\rm BI}}{\bf g}_{\rm BI},~~ {\bf h}^H_{\rm IU}=\sqrt{\zeta_{\rm IU}}{\bf g}_{\rm IU}^H,
\end{align}
where $\zeta_{\rm BI}\triangleq\epsilon d_{\rm BI}^{-\alpha_{\rm BI}}$ and $\zeta_{\rm IU}\triangleq\epsilon d_{\rm IU}^{-\alpha_{\rm IU}}$ denotes the BS$\to$IRS and IRS$\to$UE link path-loss, respectively, with $d_{\rm BI}$ ($d_{\rm IU}$) being  the link distance and $\alpha_{\rm BI} (\alpha_{\rm IU})$ being the path-loss exponent\footnote{For ease of notation, we simply use $\alpha$ to represent the path-loss exponent in the sequel for each individual link without causing confusion. }. Moreover, ${\bf g}_{\rm BI}$ (${\bf g}_{\rm IU}^H$) denotes the corresponding small-scale fading channel with $| g_{{\rm BI}, n}|$ and  $|g_{{\rm IU}, n}|$, respectively.

\begin{figure}[t!]
	\centering
	\includegraphics[height=7cm,width=10cm]{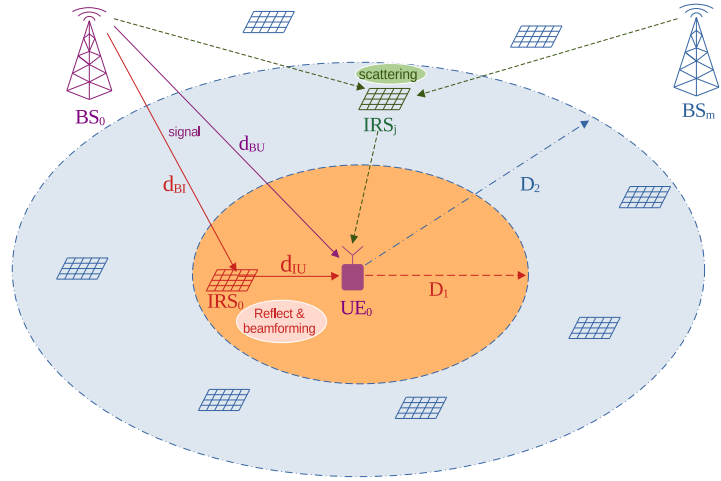}
	\caption{IRS-assisted muti-cell wireless network (DL)}
	\label{fig:IRS-assisted muti-cell wireless network (DL)}
\end{figure}

\makeatletter
\def\thickhline{%
	\noalign{\ifnum0=`}\fi\hrule \@height \thickarrayrulewidth \futurelet
	\reserved@a\@xthickhline}
\def\@xthickhline{\ifx\reserved@a\thickhline
	\vskip\doublerulesep
	\vskip-\thickarrayrulewidth
	\fi
	\ifnum0=`{\fi}}
\makeatother

\newlength{\thickarrayrulewidth}
\setlength{\thickarrayrulewidth}{1.2pt}

\begin{table}[t!]
	\caption{COMMON PARAMETERS}
	\begin{tabular}{m{0.24\linewidth}<{\centering}|m{0.62\linewidth}<{\centering}}%
		\thickhline%
		Parameter & Description\\%
		\hline%
		$ \Lambda_{\rm B}, \Lambda_{\rm I},\Lambda_{\rm U}$&Point processes of BSs, IRSs, and UEs\\
		\hline
		$ \Lambda_{\rm I,S}, \Lambda_{\rm I,F}, \Lambda_{\rm I,N} $&Daughter point processes of $\Lambda_{\rm I}$ and represent serving IRSs, interference IRSs, and noise IRSs\\
		\hline%
		$N$ & The number of IRS elements on each IRS\\%
		\hline%
		$ D_{1}, D_{2} $ & The radius of IRS serving area and interference area\\%
		\hline%
		$ h_{\rm BU} $, $ {\bf h}_{\rm BI} $, $ {\bf h}^H_{\rm IU} $&  The channel of BS$\rightarrow$UE, BS$\rightarrow$IRS, IRS$\rightarrow$UE links\\ %
		\hline%
		$g_{\rm BU}$, ${\bf g}_{\rm BI}$, ${\bf g}_{\rm IU}^H$& Small-scale fading channel of BS$\rightarrow$UE, BS$\rightarrow$IRS, IRS$\rightarrow$UE links\\%
		\hline%
		$ \zeta_{\rm BU}, \zeta_{\rm BI}, \zeta_{\rm IU} $& Path loss of BS$\rightarrow$UE, BS$\rightarrow$IRS, IRS$\rightarrow$UE links\\%
		\hline%
		$\epsilon$ & The reference channel power gain at a distance of $1$ m\\%
		\hline%
		$ H_{\rm BU}, H_{\rm BIU}, H_{\rm S} $& Channel gain of direct path, cascaded path, and the mixture channel\\%
		\hline
		$ M_{1} $, $ M_{2} $, $ I $& The number of mixture Gamma terms \\
		\hline
		$ ( \varepsilon_{i}, \beta_{i}, \xi_{i} ) $& The parameters of mixture Gamma distributions \\
		\hline
		${\mathbf{\underbar{m}}}$, ${\mathbf{\underbar{q}}}$& The index of sufficiency condition expression of multipliability and quadratic form  \\
		\hline
		$ m_{\rm BU} $, $ m_{\rm BI} $, $ m_{\rm IU} $& The fading parameters of BS$\rightarrow$UE, BS$\rightarrow$IRS, IRS$\rightarrow$UE links for Nakagami-$m$ channels \\
		\hline
		$I_{\rm F}$ & The interference   \\
		\hline
		$ K_v(y) $ &  Modified Bessel function of the second kind \\
		\hline
		$ G_{m,n}^{p,q}\left( y| a,b\right) $ &  Meijer-G function \\
		\hline
		$L_n(t)$ & Laguerre polynomial  \\
		\hline
		$t_{i}$ & The $i$-th zero of Laguerre polynomials \\
		\hline
		$\varpi_{i}$ & The $i$-th weight factor of Laguerre polynomials  \\
		\hline
		$\omega_{i}$ & The $i$-th weight factor of the $ i $-th Gamma component  \\
		\hline
		$f_{i}(x)$ & PDF of the $ i $-th Gamma component  \\
		\hline
		$\gamma\left(\cdot, \cdot \right)$ & The incomplete Gamma function \\
		\hline
		$\Gamma\left(\cdot \right)$ & The Gamma function \\
		\thickhline%
	\end{tabular}%
\label{Tab:notations}
\end{table}

For the passive IRS, let ${\bf \Theta}\triangleq{\rm diag}(e^{j\phi_{1}}, \cdots, e^{j\phi_{N}}) \in\mathbb{C}^{N\times N}$ denote its reflection matrix, where $\phi_{n}$ is the phase shift at each element $n\in\mathcal{N}\triangleq\{1, \cdots, N\}$. For the connected ${\rm IRS}_{0}$, its phase shift is adjusted to align with the direct link based on the full CSI obtained, which is given by
\begin{equation}
	[{\bf \Phi}^*]_n = e^{j(\angle{h_{\rm BU}}-\angle{[{\bf h}_{{\rm I} {\rm U}}^H]_n}-\angle{[{\bf h}_{{\rm BI}}]_n})},\forall n.
\end{equation}
However, for those not connected but interfering IRSs, the phase shift is not specifically designed and the IRSs randomly scatter received signals to ${\rm UE}_{0}$.

The received SINR is defined as below
\begin{equation}
    {\rm SINR} = \frac{\mathbf{S}}{\mathbf{I}+\delta^2},
\end{equation}
where $\mathbf{S}, \mathbf{I}, \delta^2$ represents the resived signal power, aggregated interference power, and noise power, respectively. For the given network model, the received signal at ${\rm UE}_{0}$, denoted as $y$, can be generated in four different forms. First, when the ${\Lambda}_{\rm I, {S}}$ and ${\Lambda}_{\rm I, {F}}$ are empty, the received signal only comes from the direct link and we denote this case as $y^{(1)}$. Second, when the ${\Lambda}_{\rm I, {S}}$ is empty and ${\Lambda}_{\rm I, {F}}$ is not empty, the received signal comes from the direct link with interference from ${\Lambda}_{\rm I, {F}}$ and we denote this case as $y^{(2)}$. Third, when both ${\Lambda}_{\rm I, {S}}$ and ${\Lambda}_{\rm I, {F}}$ are not empty, the received signal comes from both the direct link and IRS link with interference from ${\Lambda}_{\rm I, {F}}$ and we denote this case as $y^{(3)}$. Fourth, when the direct link is blocked, the signal can only be transmitted to UE through the IRS with interference from ${\Lambda}_{\rm I, {F}}$ and we denote this case as $y^{(4)}$. We summarized the received signal envelope for each cases below together with the corresponding signal power $\mathbf{S}$ and interference power $\mathbf{I}$ as described below
\begin{subequations}
    \begin{align}
        y^{(1)} = h_{\rm BU}x+n_{0},\quad    &\Rightarrow    \quad 
        \mathbf{S} = |h_{\rm BU}|^2,~
        \mathbf{I} = 0,\label{eqn:line-1} \\
        y^{(2)} = h_{\rm BU}x+\sum_{j\in\Lambda_{\rm I,F}} h_{\rm BIU}^{(j)}x^{'}+n_{0},
        \quad &\Rightarrow \quad 
        \mathbf{S} = |h_{\rm BU}|^2,~
        \mathbf{I} = \sum_{j\in\Lambda_{\rm I,F}} \left|h_{\rm BIU}^{(j)}\right|^2, 
        \label{eqn:line-2} \\
        y^{(3)} = (h_{\rm BU}+h_{\rm BIU}^{(0)})x+\sum_{j\in\Lambda_{\rm I,F}} h_{\rm BIU}^{(j)}x^{'}+n_{0},
        \quad &\Rightarrow \quad 
        \mathbf{S} =\left|h_{\rm BU}+h_{\rm BIU}^{(0)}\right|^2,~
        \mathbf{I} = \sum_{j\in\Lambda_{\rm I,F}} \left|h_{\rm BIU}^{(j)}\right|^2, 
        \label{eqn:line-3} \\
        y^{(4)} = h_{\rm BIU}^{(0)}x+\sum_{j\in\Lambda_{\rm I,F}} h_{\rm BIU}^{(j)}x^{'}+n_{0},
        \quad &\Rightarrow \quad 
        \mathbf{S} =\left|h_{\rm BIU}^{(0)}\right|^2,~
        \mathbf{I} = \sum_{j\in\Lambda_{\rm I,F}} \left|h_{\rm BIU}^{(j)}\right|^2, 
        \label{eqn:line-4}
    \end{align}
    \end{subequations}
where $x$ is the transmitted signal with unit transmit power $P_{\rm t}$, $x^{'}$ is the interference signal, $h_{\rm BIU}^{(j)}=\sum_{n=1}^N |{h}_{{\rm IU},n}| |{h}_{{\rm BI},n}|$ is the channel of ${\rm BS}_{0} \rightarrow{{\rm IRS}_{j}} \rightarrow{{\rm UE}_{0}}$, and $n_{0}$ is the received noise with power $\delta^2$.

\section{Mixture Gamma Approximation of Fading Channels} \label{sectionMG}

In this section, we use properties of mixture Gamma distribution to model the cascaded channel ${\rm BS}_{0} \rightarrow {\rm IRS}_{0} \rightarrow {\rm UE}_{0}$ and combined channel of the direct link ${\rm BS}_{0} \rightarrow {\rm UE}_{0}$ and cascaded link. 

\subsection{Summary of mixture Gamma distribution}

\textbf{Necessity}: In \cite{devore1993constructive}, it is proved that an arbitrary function $f(x)$ with a positive domain $x \in (0, \infty)$ and $lim_{x \to +\infty}f(x) \to 0$, can be accurately approximated as a weighted sum of Gamma distribution as written in (\ref{con:MGapprox}). Given that $f(x)$ is a valid PDF, we refer to (\ref{con:MGapprox}) as the mixture Gamma distribution with parameters $\left(\varepsilon_{i}, \beta_{i}, \xi_{i}\right)$, denoted as $f(x) \sim \mathrm{GM}\left(\varepsilon_{i}, \beta_{i}, \xi_{i}\right)$
\begin{equation}
	\begin{split}
	f(x) &= \sum_{i=1}^{\infty} \omega_{i}f_{i}(x) = \sum_{i=1}^{\infty}\varepsilon_{i}x^{\beta_{i}-1}e^{-\xi_{i}x}	\simeq \sum_{i=1}^{I}\varepsilon_{i}x^{\beta_{i}-1}e^{-\xi_{i}x},	
	\end{split}
	\label{con:MGapprox}
\end{equation}
where $f_{i}(x)=\frac{\xi_{i}^{\beta_{i}}x^{\beta_{i}-1}e^{-\xi_{i}x}}{\Gamma(\beta_{i})}$ is the PDF of a Gamma distribution with parameters $\left(\xi_i, \beta_{i}\right)$, $\Gamma(\cdot)$ is the Gamma function, $\omega_{i}=\varepsilon_{i}\cdot\Gamma(\beta_{i})/\xi_{i}^{\beta_{i}}$ is the weight of the $i$-th term, $I$ is the truncation limit that determines the approximation accuracy, and $\int_{0}^{\infty} f(x) {\rm d}x = 1$ with $f(x) \geq 0$ and $\sum_{i=1}^{\infty} \omega_i = 1$.

\textbf{Sufficiency}: In \cite{devore1993constructive}, the existence of a mixture Gamma function $S_u\left(x\right)$ that uniformly converges to an arbitrary function $f(x)$ is proved as written below
\begin{equation}
    \begin{split}
        &lim_{u \to +\infty} S_u(x) = f(x)  \quad \text{uniformly for } 0 < x < \infty,\\
        \text{where } S_u(x) &= \sum_{k=0}^{\infty} \frac{1}{u}f\left(\frac{k}{u}\right) \mathrm{Gamma}\left(u, k+1\right) = \sum_{k=0}^{\infty} \frac{1}{u}f\left(\frac{k}{u}\right) \cdot \frac{u^{k+1} x^k}{k!} e^{-u x},
    \end{split}
    \label{eq-11-rev}
\end{equation}
and $u$ is an arbitrarily large number that determines the approximation accuracy. 
The equality in (\ref{eq-11-rev}) indicates that an arbitrary function $f(x)$ can be accurately approximated by a mixture of Gamma distributions $\mathrm{Gamma}\left(u, k+1\right)$ with parameters $u$, $k+1$ and weight $\frac{1}{u}f\left(\frac{k}{u}\right)$. Note that (\ref{con:MGapprox}) represents the necessity condition to construct an arbitrary function from a mixture of Gamma distributions, whereas (\ref{eq-11-rev}) corresponds to the sufficiency condition that maps the weight $\omega_i$ and $f_i(x)$. We can find a direct relation between the arbitrary distribution $f(x)$ and tuples $\left(\varepsilon_{i}, \beta_{i}, \xi_{i}\right)$ by using (\ref{con:MGapprox}) and (\ref{eq-11-rev}) as described below
\begin{equation}
	\left(\varepsilon_{i}, \beta_{i}, \xi_{i}\right) = \left(\frac{u^{i-1}}{\Gamma(i)} \cdot f\left(\frac{i-1}{u}\right), i, u  \right).
	    \label{eq-9-rev}
\end{equation}

The statistics of mixture Gamma distribution, including the CDF, moments, and Laplace transform of a mixture Gamma distributed random variable, are derived in \cite{chun2017comprehensive} as follows 
\begin{equation}
    \begin{split}
        F(x) = \sum_{i=1}^{I} \varepsilon_{i}\xi^{-\beta_{i}}\gamma(\beta_{i}, \xi_{i}x),\quad 
        \mathbb{E}\left[ x^l  \right] = \sum_{i=1}^{I}  \varepsilon_{i} \frac{\Gamma(\beta_i+l)}{\xi_i^{\beta_i+l}}, \quad 
        \mathcal{L}(s) = \sum_{i=1}^{I}  \varepsilon_{i} \frac{\Gamma(\beta_i)}{(\xi_i+s)^{\beta_i}},
    \end{split}
\end{equation}
where $\gamma\left(\cdot, \cdot \right)$ is the incomplete Gamma function. The authors in \cite{6059452} proved that majority of the known fading models can be approximated by the mixture Gamma distribution. Particularly, Rayleigh and Nakagami-$m$ fading can be represented by a mixture Gamma distribution with a single term. For an arbitrary fading model, whose PDF can be approximated by a mixture Gamma distribution, $I$ is no need to be larger than 20 with accuracy less than $10^{-5}$ \cite{6059452}. 

However, the IRS-assisted network in Fig. \ref{fig:IRS-assisted muti-cell wireless network (DL)} involves combined channel that is composed of the double-faded, cascaded link through IRS and the direct link. Traditional works on cascaded channel utilized specialized function, such as Fox-H or Meijer G-function, which lack tractability and are hard to gain any insights. To resolve this issue, we adopt the mixture Gamma approximation to model the wireless channels of IRS-assisted network.

\subsection{Properties of mixture Gamma distribution}

\begin{theorem} \label{theorem:Multipliability} \textbf{Multipliability:}
The product distribution of two independent, mixture Gamma distributed random variables $X_{1}$ and $X_{2}$ can be represented by a mixture Gamma distribution $Y=X_{1}X_{2}$ with parameters $\left(\varepsilon_{\mathbf{\underbar{m}}}, \beta_{\mathbf{\underbar{m}}}, \xi_{\mathbf{\underbar{m}}}\right)$ as described below
\begin{small}
    \begin{equation}
        \begin{split}
            &\text{Given } 
            f_{\rm X_1}(x_1) = \sum_{m_{1}=0}^{M_{1}} \omega_{m_{1}} \cdot \mathrm{Gamma}\left(\varepsilon_{m_{1}}, \beta_{m_{1}}\right) \text{ and }
            f_{\rm X_2}(x_2) =\sum_{{m_{2}}=0}^{{M_{2}}} \omega_{m_{2}} \cdot \mathrm{Gamma}\left(\varepsilon_{m_{2}}, \beta_{m_{2}}\right),\\
            &\text{the product distribution of } Y=X_{1}X_{2} \text{ follows }
f_{\rm Y}(y) =  \sum_{\mathcal{C}_{\mathbf{\underbar{m}}}} \varepsilon_{\mathbf{\underbar{m}}} \cdot y^{\beta_{\mathbf{\underbar{m}}}-1} \cdot e^{-y \cdot \xi_{\mathbf{\underbar{m}} }},
        \end{split}
        \label{MGdouble}
    \end{equation}	
\end{small}%
where the summation range $\mathcal{C}_{\mathbf{\underbar{m}}}$ and parameters are defined as  
\begin{equation}
    \begin{split}
  \mathcal{C}_{\mathbf{\underbar{m}}} &= \{0 \leq m_{1} \leq M_1,~ 0 \leq m_2 \leq M_2,~ 1 \leq i \leq I\},\\
  \beta_{\mathbf{\underbar{m}}} &= \varepsilon_{m_{1}},\quad \xi_{\mathbf{\underbar{m}}} = \frac{\beta_{m_{1}}\beta_{m_{2}}}{t_{i}}, \quad  
\varepsilon_{\mathbf{\underbar{m}}} =  \left(\prod_{j=1}^{2} \frac{\omega_{m_j} \beta_{m_j}^{\varepsilon_{m_1}} }{\Gamma\left(\varepsilon_{m_j}\right) } \right) 
\cdot \varpi_{i}{t_{i}^{-\varepsilon_{m_{1}}+\varepsilon_{m_{2}}-1}},
    \end{split}
\end{equation}
$t_{i}$ is the $i$-th root of the Laguerre polynomial $L_n(t)$ and $\varpi_{i}$ is the $i$-th weight of the Gaussian-Laguerre quadrature $\int_0^\infty e^{-t}f(t)dt \approx \sum_{i=1}^{n} \varpi_{i} f\left(t_i\right)$ defined as $\varpi_{i} = \frac{t_i}{\left(n+1\right)^2 L_{n+1}\left(t_i\right)^2}$ \cite{abramowitz1964handbook}.
\end{theorem}

\begin{proof}
    See Appendix \ref{appendix:CascadedMG}.
\end{proof}

Given two independent Gamma distributed random variables $X_1$ and $X_2$, the product distribution of $Y=X_{1}X_{2}$ can be further simplified by substituting $M_1 = M_2 = 0$, $\omega_{m_1} = \omega_{m_2} = 1$ in (\ref{MGdouble}), which is described in the following lemma. 

\begin{lemma}\label{Lem:GamProduct}
The product distribution of two independent Gamma-distributed random variables $X_1$ and $X_2$ is given by 
    \begin{equation}
        \begin{split}
f_{\rm Y}(y) &=  \sum_{i=1}^{I} \varepsilon_{\mathbf{\underbar{m}}} \cdot y^{\beta_{\mathbf{\underbar{m}}}-1} \cdot e^{-y\cdot\xi_{\mathbf{\underbar{m}} }},
        \end{split}
        \label{MGdouble-lemma1-eq1}
    \end{equation}	
where $X_1 \sim \mathrm{Gamma}\left(\varepsilon_{{m_{1}}}, \beta_{{m_{1}}}\right)$, $X_2 \sim \mathrm{Gamma}\left(\varepsilon_{{m_{2}}}, \beta_{{m_{2}}}\right)$, and the parameter tuple $\left(\varepsilon_{\mathbf{\underbar{m}}}, \beta_{\mathbf{\underbar{m}}}, \xi_{\mathbf{\underbar{m}}}\right)$ is defined as follows 
\begin{equation}
    \begin{split}
  \beta_{\mathbf{\underbar{m}}} = \varepsilon_{m_{1}},\quad \xi_{\mathbf{\underbar{m}}} = \frac{\beta_{m_{1}}\beta_{m_{2}}}{t_{i}}, 
  \quad 
\varepsilon_{\mathbf{\underbar{m}}} =  \left(\prod_{j=1}^{2} \frac{\beta_{m_j}^{\varepsilon_{m_1}} }{\Gamma\left(\varepsilon_{m_j}\right) } \right) 
\cdot \varpi_{i}{t_{i}^{-\varepsilon_{m_{1}}+\varepsilon_{m_{2}}-1}}.
    \end{split}
\end{equation}
\end{lemma}

\begin{theorem}\label{theorem:Mixture}\textbf{Quadratic form:}
Given two independent, mixture Gamma distributed random variables $X^2$ and $Y^2$, the quadratic form $S = (X+Y)^{2}$ follows a mixture Gamma distribution with parameters $\left(\varepsilon_{\mathbf{\underbar{q}}}, \beta_{\mathbf{\underbar{q}}}, \xi_{\mathbf{\underbar{q}}}\right)$
    \begin{equation}
        \begin{split}
            &f_{\rm S}(s) = \sum_{\mathcal{C}_{\mathbf{\underbar{q}}}} \left( \varepsilon_{\mathbf{\underbar{q}_1}} 
            e^{-s \cdot \xi_{\mathbf{\underbar{q}_1} }} - \varepsilon_{\mathbf{\underbar{q}_2}} 
            e^{-s \cdot \xi_{\mathbf{\underbar{q}_2} }} \right)
            \cdot s^{\beta_{\mathbf{\underbar{q}}}-1},\\
            &\text{where } 
            X \sim \mathrm{GM}\left(\varepsilon_{m_1}, \beta_{m_1}, \xi_{m_1}\right), \text{ and }  Y \sim \mathrm{GM}\left(\varepsilon_{m_2}, \beta_{m_2}, \xi_{m_2}\right),
    \end{split}
    \label{thm2-eq17-rev}
\end{equation}
the summation range $\mathcal{C}_{\mathbf{\underbar{q}}}$ and parameter tuples are given by 
\begin{small}
\begin{equation}
    \begin{split}
  \mathcal{C}_{\mathbf{\underbar{q}}} &= \{0 \leq m_{1} \leq M_1,~ 0 \leq m_2 \leq M_2, 0 \leq k_1 \leq 2\beta_{m_2}-1,\\
 &\quad \quad  0 \leq k_2 \leq 2\beta_{m_1}-1+k_1,~ 0 \leq k_3 \leq \infty\},
    \end{split}
    \label{thm2-eq18a-rev}
\end{equation}
\end{small}%
\begin{equation}
    \begin{split}
\beta_{\mathbf{\underbar{q}}} &= \beta_{m_1}+\beta_{m_2}+k_{3}, \quad 
\xi_{\mathbf{\underbar{q}_1}} = \xi_{m_2}, \quad 
\xi_{\mathbf{\underbar{q}_2}} = \xi_{m_1},\\
\varepsilon_{\mathbf{\underbar{q}_1}} &= (-1)^{k_{2}} \chi \cdot \xi_{m_2}^{k_{2}+2k_{3}+1}, \quad  	
\varepsilon_{\mathbf{\underbar{q}_2}} =	\chi \cdot \xi_{m_1}^{k_{2}+2k_{3}+1},    
\end{split}
    \label{thm2-eq18b-rev}
\end{equation}
and $\chi=\frac{\varepsilon_{i}\varepsilon_{j}\tbinom{2\beta_{j}-1}{k_{1}}\tbinom{2\beta_{i}-1+k_{1}}{k_{2}}(-1)^{k_{1}}\xi_{j}^{2\beta_{i}+k_{1}-k_{2}-1}\Gamma(\frac{k_{2}+1}{2})} {\Gamma(\frac{k_{2}+1}{2}+k_{3}+1)(\xi_{i}+\xi_{j})^{2\beta_{i}+k_{1}+k_{3}}}$.
\end{theorem}
\begin{proof}
See Appendix~\ref{appendix:TheoremMixtureChannel}.
\end{proof}

Although the mixture Gamma distributions in (\ref{MGdouble}) and (\ref{thm2-eq17-rev}) involve multiple summations, it is worth noting that these functions still converge extremely fast, thanks to the rapid convergence of the weight terms. In Fig.~\ref{fig:MG_mmse}, we validated that the mixture Gamma distributions achieve an approximation error of less than $10^{-4}$ with only ten terms. 

\begin{remark}
	Since the production distribution of two independent, mixture Gamma random variables is still a mixture Gamma, Theorem~\ref{theorem:Multipliability} can be easily extended to a  multiplication of $K$ independent, mixture Gamma distributed random variables. We introduced a heuristic algorithm in Appendix~\ref{appendix:IterationAlgorithmForChannel} to evaluate the product distribution of $K$ independent, mixture Gamma random variables. Similarly, the distribution of the quadratic form can be easily extended to $K$ independent, mixture Gamma random variables. Hence, the analytical framework derived in this paper can be applied to network environments with multiple IRS association. 
\end{remark}

Some mixture Gamma approximations of single links and cascaded links are provided in Fig.~\ref{fig:table} for ease reference.


\section{Performance analysis} \label{section:Perf}

\subsection{Channel power statics}
For network performance analysis, we assume that the transmit power is one and the amplitude $g_{\rm BU}$, $g_{{\rm BI},n}$ and $g_{{\rm IU},n}$ follow Nakagami-$m$ distribution with parameters $m_{\rm BU}$, $m_{\rm BI}$ and $m_{\rm IU}$, respectively. Let us denote the power terms as follows
\begin{align}
	H_{\rm BU}\triangleq |h_{\rm BU}|^2=\epsilon d_{\rm BU}^{-\alpha} |g_{\rm BU}|^2, \label{Eq:HBU}
\end{align} 
\begin{align}
	H_{\rm BIU}=  \left|{\bf h}^H_{\rm IU}{\bf h}_{\rm BI}\right|^{2} \triangleq |h_{\rm BIU}|^2=\epsilon^2 d_{\rm BI}^{-\alpha}  d_{\rm IU}^{-\alpha}  \left|\sum_{n=1}^N  |{g}_{{\rm IU},n}| |{g}_{{\rm BI},n}|\right|^2,\label{Eq:HBIU}
\end{align} 
\begin{align}
	H_{\rm S}\triangleq \left||h_{\rm BU}|+|h_{\rm BIU}|\right|^2. \label{Eq:HS}
\end{align} 

\subsubsection{Single path}

As $g_{\rm BU}$, $g_{\rm BI}$ and $g_{\rm IU}$ are distributed by the Nakagami-$m$ fading, the power term $H_{\rm BU}$ follows a Gamma distribution, whereas the statistics of $H_{\rm BIU}$ is characterized by the mixture Gamma distribution and the parameter tuples are described in the following lemmas. 

\begin{lemma}\label{Lem:H_BU}$H_{\rm BU}$ follows the Gamma distribution, which can be modeled as a mixture Gamma distribution with $I = 1$ and
		\begin{align}
		\left(\varepsilon_{\rm BU}, \beta_{\rm BU}, \xi_{\rm BU}\right) = \left(\frac{({d_{\rm BU}^{\alpha }m_{\rm BU}  })^{m_{\rm BU}  }}{\epsilon^{m_{\rm BU}  }\Gamma(m_{\rm BU}  )}, m_{\rm BU},  \frac{m_{\rm BU} d_{\rm BU}^{\alpha}}{\epsilon}\right).
	\end{align} 
\end{lemma}

\begin{lemma}\label{Lem:HBIU}By Theorem~\ref{Lem:GamProduct}, $H_{\rm BIU}$ follows the mixture Gamma distribution as (\ref{MGdouble}) with parameters $\left(\varepsilon, \beta, \xi\right)$ given by 
		\begin{align}\label{eq: MG_para_BIU}
			\begin{split}
				&~ \varepsilon_{{\rm BIU},i} = \frac{(m_{\rm BI}m_{\rm IU})^{m_{\rm BI}  }\varpi_{i}t_{i}^{m_{\rm IU}  -m_{\rm BI}  -1}}{\Gamma(m_{\rm BI}  )\Gamma(m_{\rm IU}  )}\left({\frac{W}{N^2}}\right)^{m_{\rm BI}  }, \\ &~
				\beta_{{\rm BIU},i} = m_{\rm BI}  , \quad
				\xi_{{\rm BIU},i} = \frac{m_{\rm BI}  m_{\rm IU}  }{t_{i}}{\frac{W}{N^2}},
			\end{split}
		\end{align} 
		where $W= \frac{{d^{\alpha}_{\rm BI}}{d^{\alpha}_{\rm IU}}}{\epsilon^2}$ and $I=20$ achieves sufficient approximation error of less than $10^{-5}$.
\end{lemma}

\subsubsection{Mixture path}
In the following Lemma, we use Theorem~\ref{theorem:Mixture} to characterize the distribution of the combined channel, given that the channel gain of the individual paths follows a mixture Gamma distribution. 
\begin{lemma}
Given that the channel gain of the direct link and cascaded link follow mixture Gamma distributions, the channel gain of the combined channel, $H_{\rm S}$, follows a mixture Gamma distribution as (\ref{thm2-eq17-rev}) with parameters given by 
	\begin{small}
\begin{equation}
    \begin{split}
  \mathcal{C}_{\mathbf{\underbar{q}}} &= \{1 \leq m_{1} \leq M_1,~ m_2 = 1, 0 \leq k_1 \leq 2m_{\rm BU}-1,\\
 &\quad \quad  0 \leq k_2 \leq 2m_{\rm BI}+k_1-1,~ 0 \leq k_3 \leq \infty\},\\
\beta_{\mathbf{\underbar{q}}} &= \beta_{{\rm BIU},i}+\beta_{\rm BU}+k_{3}, \quad 
\xi_{\mathbf{\underbar{q}_1}} = \xi_{\rm BU}, \quad 
\xi_{\mathbf{\underbar{q}_2}} = \xi_{{\rm BIU},i},\\
\varepsilon_{\mathbf{\underbar{q}_1}} &= (-1)^{k_{2}} \chi \cdot \xi_{\rm BU}^{k_{2}+2k_{3}+1}, \quad  	
\varepsilon_{\mathbf{\underbar{q}_2}} =	\chi \cdot \xi_{{\rm BIU},i}^{k_{2}+2k_{3}+1},    \\
\end{split}
\end{equation}
\end{small}%
	\begin{small}
	\begin{equation}
		\begin{split}
			\chi &= \frac{\varepsilon_{\rm BU}\varepsilon_{\rm BIU,i}\tbinom{2\beta_{\rm BU}-1}{k_{1}}(-1)^{k_{1}}\tbinom{2\beta_{\rm BU}-k_{1}-1}{k_{2}}\xi_{\rm BU}^{2\beta_{\rm BU}+k_{1}-k_{2}-1}\Gamma(\frac{k_{2}+1}{2}) } {\Gamma(\frac{k_{2}+1}{2}+k_{3}+1)(\xi_{{\rm BIU},i}+\xi_{\rm BU})^{2\beta_{\rm BU}+k_{1}+k_{3}}}.
    \end{split}
\end{equation}
\end{small}%

\end{lemma}

\subsection{Laplace transform of the aggregated interference power}

In this subsection, we will derive the Laplace transform of the aggregated interference. The interference power received at the ${\rm UE}_{0}$ from direct links and cascaded links are given by $I_{{\rm F},1}$ and $I_{{\rm F},2}$, respectively
\begin{equation}
	\begin{split}
		I_{{\rm F},1} = &~ \sum_{m \in \Lambda_{\rm B} \setminus\{0\}}  H_{{\rm BU}}^{(m)}  ,\\ I_{{\rm F},2} = &~ \sum_{m \in \Lambda_{\rm B} \setminus\{0\}} \sum_{j\in \Lambda_{\rm I,F}\bigcup\Lambda_{\rm I,S}}  H_{\rm BIU}^{(m,j)} . 
	\end{split}  
\end{equation}
Laplace transforms of the interference power are given by
\begin{align} \label{eq:L_I12}
	\mathcal{L}_ {\rm I_{{\rm F},1}} |_{d_{\rm BU}^{(0)}} =&~  \mathbb{E}[e^{-sI_{{\rm F},1}}] |_{d_{\rm BU}^{(0)}}  =   \mathbb{E}_{\Lambda_{B}\setminus\{0\}}\left\{ e^{-s H_{\rm BU} } \right\}\Big|_{d_{\rm BU}^{(0)}} , \notag \\ 
	\mathcal{L}_{\rm I_{{\rm F},2}}| _{d_{\rm BIU}^{(0)}}= &~  \mathbb{E}_{\Lambda _{B}\setminus\{0\},\Lambda_{\rm I,F}\bigcup\Lambda_{\rm I,S}}\left\{ e^{-s \eta H_{\rm BIU}} \right\}\Big|_{d_{\rm BU}^{(0)}} ,
\end{align}
where $d_{\rm BU}^{(m)}$ is the distance from $ {\rm BS}_{m}$ to $ {\rm UE}_{0}$, $\eta = \mathbb{E}\left[{d_{\rm BI}^{(m,j)}}^{-\alpha}\right]$, $d_{\rm BI}^{(m,j)}$ is the distance from ${\rm BS}_{m}$ to ${\rm IRS}_{j}$, and we assumed $d_{\rm BU} \simeq d_{\rm BI}$.

If $\Lambda_{\rm I,F}$ is empty, the aggregated interference and its corresponding Laplace transform are given by 
\begin{equation}\label{eq:L_IFo}
    \begin{split}
        I_{{\rm F},\rm o} &= I_{{\rm F},1}, \quad 
        \mathcal{L}_{ \rm I_{{\rm F},\rm o}} |_{d_{\rm BU}^{(0)}} = \mathbb{E}[e^{-sI_{{\rm F},1}}] |_{d_{\rm BU}^{(0)}}  =  \mathcal{L}_{\rm I_{ {{\rm F},1}}} |_{d_{\rm BU}^{(0)}}.        
    \end{split}
\end{equation}
If $\Lambda_{\rm I,F}$ is not empty, the aggregated interference and its corresponding Laplace transform are given by 
\begin{equation}\label{eq:L_IFa}
        I_{{\rm F},\rm a} = I_{{\rm F},1} + I_{{\rm F},2}, \quad
        \mathcal{L}_{ \rm I_{{\rm F},\rm a}} |_{d_{\rm BU}^{(0)}} = \mathcal{L}_ {\rm I_{{\rm F},1}} |_{d_{\rm BU}^{(0)}} \cdot\mathcal{L}_ {\rm I_{{\rm F},2}} |_{d_{\rm BU}^{(0)}} . 
\end{equation}
The closed form expressions of (\ref{eq:L_IFo}) and \eqref{eq:L_IFa} are given in (\ref{eq:L_I1_appen}), (\ref{eq:L_I2_appen}) and derived in Appendix \ref{appendix:LaplaceI}. 

The CDF of the aggregated interference power can be numerically obtained by taking the inverse Laplace transform of $\mathcal{L}_{ \rm I|_{d_{\rm BU}^{(0)}}}(s)$ and MATLAB offers software library to evaluate the operation as follows
\begin{equation}
	F_{ \rm I_{\rm F}|_{d_{\rm BU}^{(0)}}}(x) = \mathcal{L}^{-1}\left[ \frac{1}{s} \mathcal{L}_{ \rm I_{\rm F}|_{d_{\rm BU}^{(0)}}}(s) \right] (x). \label{LaplaceICDF} 
\end{equation}
 



\subsection{Analytical framework}

In this subsection, we adopt an analytical framework to assess the system performance metrics by using stochastic geometry. The original idea was proposed by Hamdi in \cite{hamdi2007useful} for Nakagami-$m$ fading, later in \cite{chun2017stochastic} for $\kappa$ - $\mu$ and $\eta$ - $\mu$ fading, and in \cite{chun2017comprehensive} for $\kappa$ - $\mu$ shadowed fading, which we further extend to IRS-assisted networks with mixture Gamma distributed channels. With the proposed method, any performance measures can be evaluated and represented as a function of SINR $g(\rm SINR)$, including the spectral efficiency, moments of SINR, and outage probability. 

\begin{theorem} \label{theoremPerfAnalysis}
	For a network with mixture Gamma distributed channels, whose received signal can be modeled as a mixture Gamma distribution with tuple $ (\varepsilon_{i},\beta_{i},\xi_{i})$, $\mathbb{E}[g(\rm SINR)]$ is given by
	\begin{equation}
			\mathbb{E}[g(\rm SINR)] =  \int_{0}^{\infty} {g(\rm SINR)f_{\rm S}(s)}\mathrm{d}s = \sum_{i=0}^{I} \varepsilon_{i} \Gamma(\beta_{i}) {\xi_{i}}^{-\beta_{i}}  \int_{0}^{\infty} {g_{\beta_{i}}(z) e^{-\delta^{2} {\xi_{i}} z}} \mathcal{L}_{\rm I_{\rm F}}({\xi_{i}} z)  \mathrm{d}z ,
			\label{eq-32-thm3-rev}
	\end{equation}
	where $g_{\beta_{i}}(z)$ is defined as
	\begin{equation}
		\begin{split}
			g_{{\beta_{i}}}(z) = \frac{1}{\Gamma({\beta_{i}})}\frac{\mathrm{d}^{{\beta_{i}}}}{\mathrm{d}z^{{\beta_{i}}}}g(z) .
		\end{split}
	\end{equation}
	\end{theorem}
	\begin{proof}
		See Appendix \ref{appendix:Perf}.
	\end{proof}
	
	In the following, we utilize Theorem \ref{theoremPerfAnalysis} and (\ref{eq-32-thm3-rev}) to evaluate several system performance metrics of interest by invoking their SINR functions.
	\subsubsection{Spectral efficiency}
	Spectral efficiency is given by \cite{jo2012heterogeneous}
	\begin{equation} \label{spectralEfficiency}
		\mathcal{R} = \mathbb{E}[\ln(1+\rm SINR)].
	\end{equation}
	By substituting $g(z)=\ln(1+z)$ and $g_{\beta_{i}}(z)$ to (\ref{eq-32-thm3-rev}) \cite{hamdi2007useful}
	\begin{equation} \label{eq:gz}
		\begin{split}
			g_{\beta_{i}}(z) =  \frac{1}{\Gamma(\beta_{i})}\frac{\mathrm{d}^{\beta_{i}}}{\mathrm{d}z^{\beta_{i}}}g(z)= \frac{1}{z}\left( 1-\frac{1}{(1+z)^{\beta_{i}}} \right),
		\end{split} 
	\end{equation}
	the spectral efficiency of an IRS-assisted wireless network is evaluated as follows 
	\begin{equation}
		\begin{split}
			\mathcal{R} =  \sum_{i=0}^{I} \varepsilon_{i} \Gamma(\beta_{i}) {\xi_{i}}^{-\beta_{i}}  \int_{0}^{\infty} { {\frac{1}{z}\left( 1-\frac{1}{(1+z)^{\beta_{i}}} \right) } \frac{\mathcal{L}_{\rm I_{\rm F}}({\xi_{i}} z)}{e^{\delta^{2}{\xi_{i}} z}} }  \mathrm{d}z .
		\end{split} 
	\end{equation}

	\subsubsection{Moments of SINR}
	
	The moments of the SINR $\mathbb{E}[{\rm SINR}^{l}]$ can be derived by substituting $g(z)=z^{l}$ and $g_{\beta_{i}}(z)$ to (\ref{eq-32-thm3-rev})
	\begin{equation}
		g_{\beta_{i}}(z) =  \frac{1}{\Gamma(\beta_{i})}\frac{\mathrm{d}^{\beta_{i}}}{\mathrm{d}z^{\beta_{i}}}g(z)=\frac{\Gamma(\beta_{i}+l)}{\Gamma(l)\Gamma(\beta_{i})}z^{l-1}.
	\end{equation}
	Then, the moments of the SINR is evaluated as follows \cite{chun2017comprehensive}
	\begin{equation}
		\begin{split}
			\mathbb{E}[{\rm SINR}^{l}] =  \sum_{i=0}^{I} \varepsilon_{i} \Gamma(\beta_{i}) {\xi_{i}}^{-\beta_{i}}  \int_{0}^{\infty} { { \frac{\Gamma(\beta_{i}+l)}{\Gamma(l)\Gamma(\beta_{i})}z^{l-1} } e^{-\delta^{2}\xi_{i} z}} \mathcal{L}_{\rm I_{\rm F}}({\xi_{i}} z)  \mathrm{d}z  .
		\end{split} 
	\end{equation}
	
	\subsubsection{Outage probability}
	
	The outage probability is defined as written below, which is averaged over the link distance
	\begin{equation}
		\begin{split}
			P_{\rm outage}= &1 -  \mathbb{P}\{{\rm SINR} > \tau \} = 1-\mathbb{E}\left[ \mathbb{P} \left( I_{\rm F} < \frac{S}{\tau}-\delta^{2} \right) \right] , \label{con:outageP}
		\end{split} 
	\end{equation}	
	for a given SINR threshold $\tau$. By substituting \eqref{LaplaceICDF} into \eqref{con:outageP}, the outage probability can be further simplied to 
	\begin{equation}
		P_{\rm outage} = 1- \mathcal{L}^{-1}\left[ \frac{1}{s} \mathcal{L}_{I_{\rm F}|_{d_{\rm BU}}^{(0)}}(s) \right] \left(\frac{S}{\tau}-\delta^{2}\right). \label{OutageLaplaceInverse} 
	\end{equation}
	The CDF of the interference can be evaluated by using the Gil-Pelaez's inversion as described below
	\begin{equation}
	    \begin{split}
				\mathbb{P}(I_{\rm F}<x) =&~ \frac{1}{2} - \frac{1}{\pi} \int_{0}^{\infty} {\frac{Im\{ e^{itx}\mathcal{L}_{\rm I_{F}}(it) \}}{t} \mathrm{d}t},\quad i = \sqrt{-1} \\ 
				=&~ \frac{1}{2}+\int_{0}^{\infty} {Im \left\{ \sum_{i=1}^{I} \varepsilon_{i} \frac{\Gamma(\beta_{i})}{(\xi_{i}-it)^{\beta_{i}}} \phi(x) \right\} \mathrm{d}t},\\
				\phi(x) \triangleq&~ \int_{0}^{\infty} {\exp\left[ i\delta^{2}\lambda_{\rm I}^{-\frac{\alpha}{2}}xt^{\frac{\alpha}{2}} - \pi t \psi(ix) \right]} \mathrm{d}t, \\
				\psi(z) \triangleq&~ \mathbb{E}_{\rm S}\left[ {}_{1}F_{1} \left[ ~{\left.{\!\genfrac..{0pt}{}{-\frac{2}{\alpha}}{1-\frac{2}{\alpha}}}~\right |zH_{\rm S}} \right] \right],  	        
	    \end{split}
	\end{equation}
	where we used \cite[eq.(4)]{di2014stochastic}.
	
	\section{NUMERICAL RESULTS} \label{sectionNumerical}
	
	\begin{figure*}[t!]
		\centering
		\subfigure[PDF of mixture Gamma approximation of the cascaded link.]{
			\begin{minipage}[t]{0.49\textwidth}
				\includegraphics[width=\linewidth]{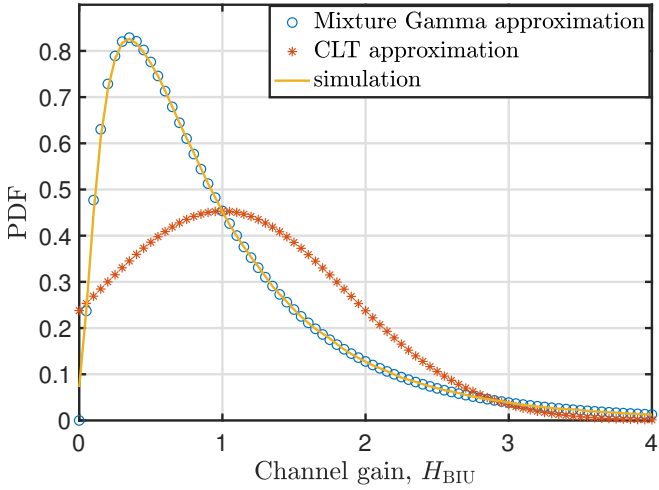}
			\end{minipage}
			\label{fig:MG_double}
		}~
		\subfigure[MMSE of mixture Gamma approximation.]{
			\begin{minipage}[t]{0.49\textwidth}
				\includegraphics[width=\linewidth]{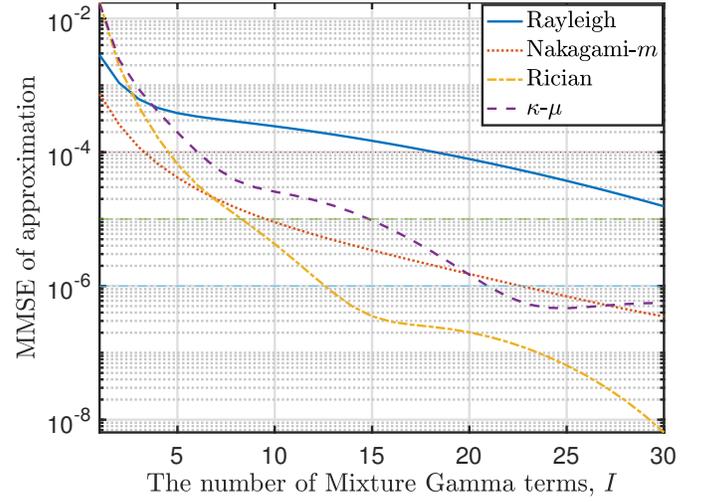}
			\end{minipage}
			\label{fig:MG_mmse}
		}
		\centering
		\subfigure[mixture Gamma approximation of the mixture channel.]{
			\begin{minipage}[t]{0.49\textwidth}
				\includegraphics[width=\linewidth]{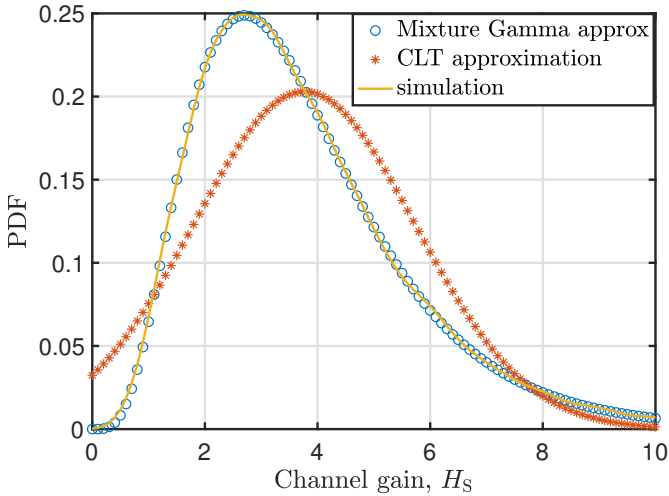}
			\end{minipage}
			\label{fig:MG_mixture}
		}~
		\subfigure[PDF of mixture Gamma approximation of all kinds of channel in IRS-assisted networks.]{
			\begin{minipage}[t]{0.49\textwidth}
				\includegraphics[width=\linewidth]{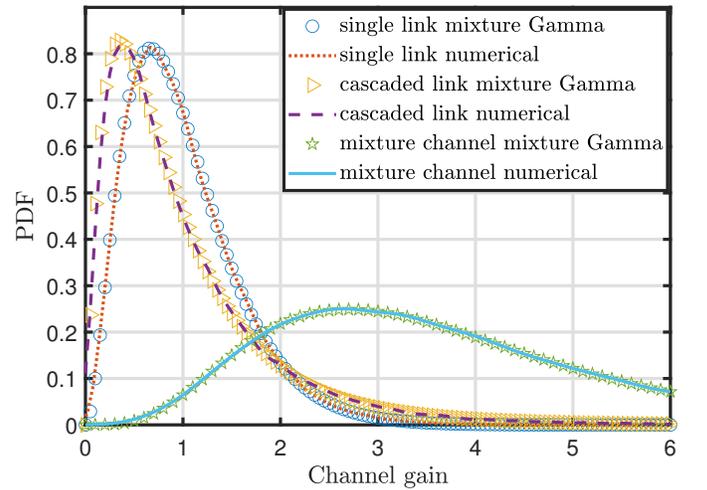}
			\end{minipage}
			\label{fig:MG_all_link}
		}
		\caption{Channel modeling of the channel gains.}
		\label{fig:MG_approx}
		\vspace{1em}
	\end{figure*}

	\begin{figure*}[t!]
		\centering
		\subfigure[Signal power CDF with different distance of IRS-UE.]{
			\begin{minipage}[t]{0.49\textwidth}
				\includegraphics[width=\linewidth]{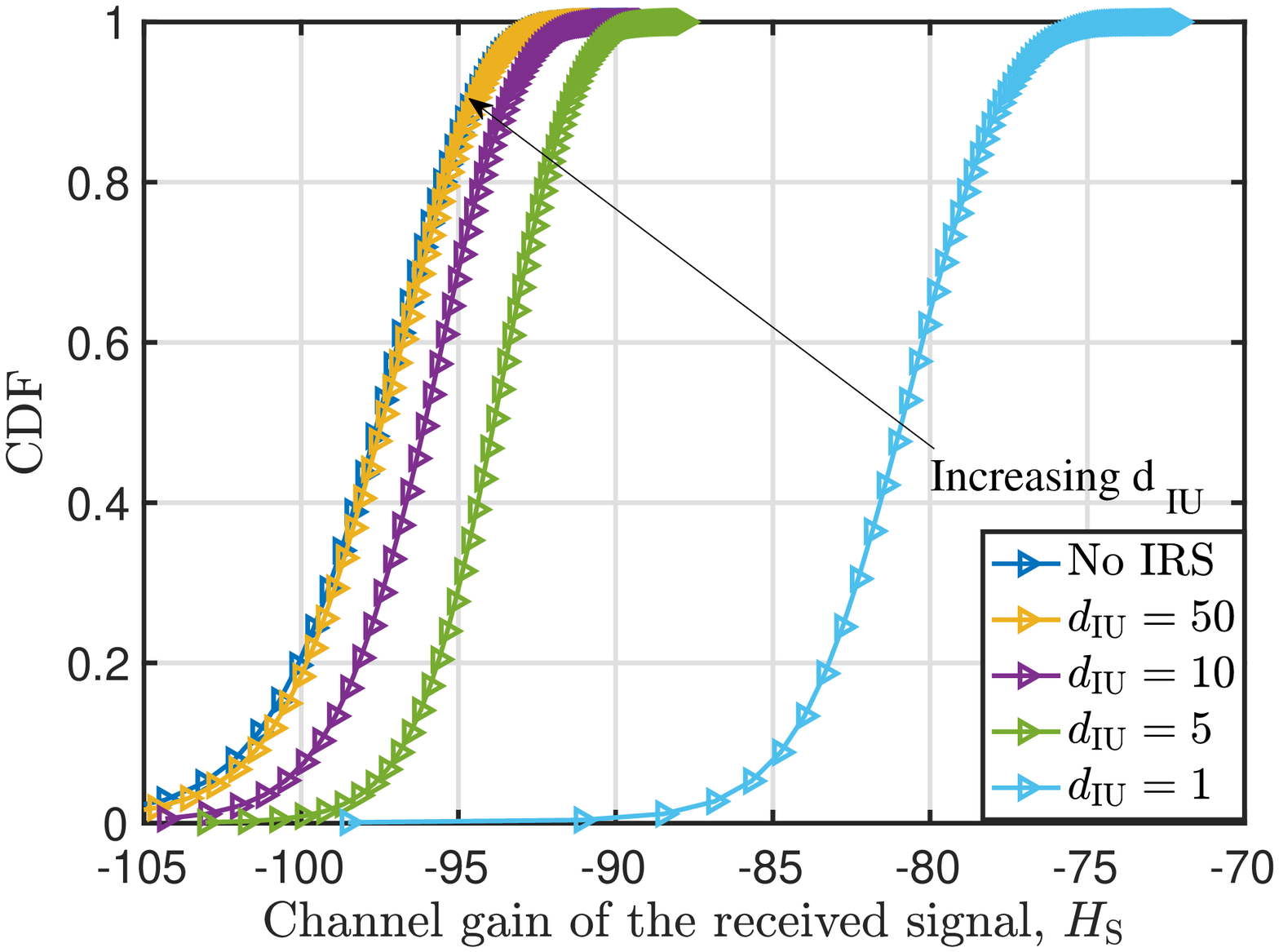}
			\end{minipage}
			\label{fig:sigP_d_IU}
		}~
		\subfigure[Tradeoff between the distance of IRS-UE and the number of IRS elements.]{
			\begin{minipage}[t]{0.49\textwidth}
				\includegraphics[width=\linewidth]{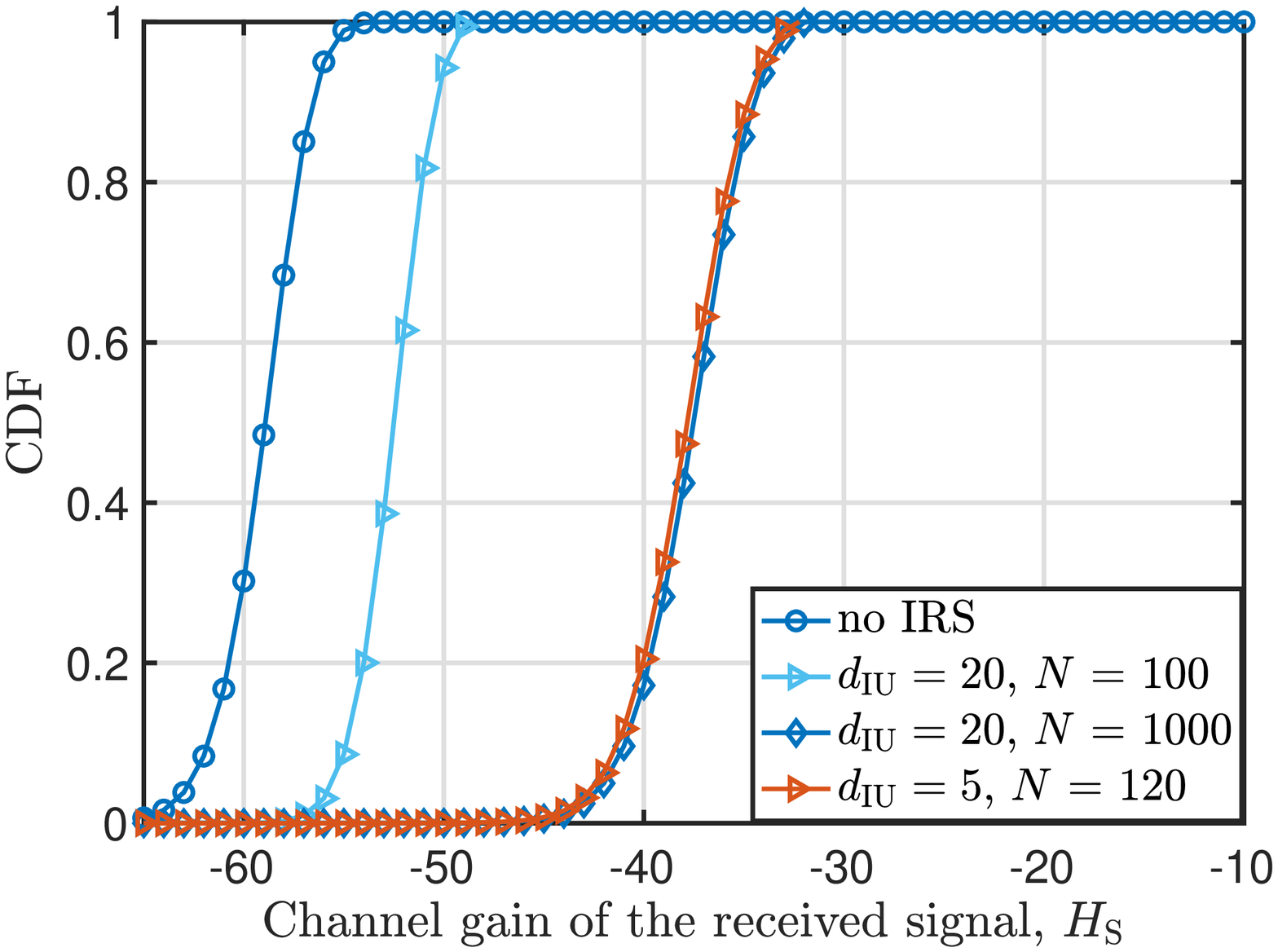}
			\end{minipage}
			\label{fig:sigP_tradeoffND}
		}
		\caption{The received signal power distribution given link distances.}
		\label{fig:sigP}
		\vspace{1em}
	\end{figure*}

	\begin{figure*}[t!]
		\centering
		\subfigure[Spectral Efficiency.]{
			\begin{minipage}[t]{0.49\textwidth}
				\includegraphics[width=\linewidth]{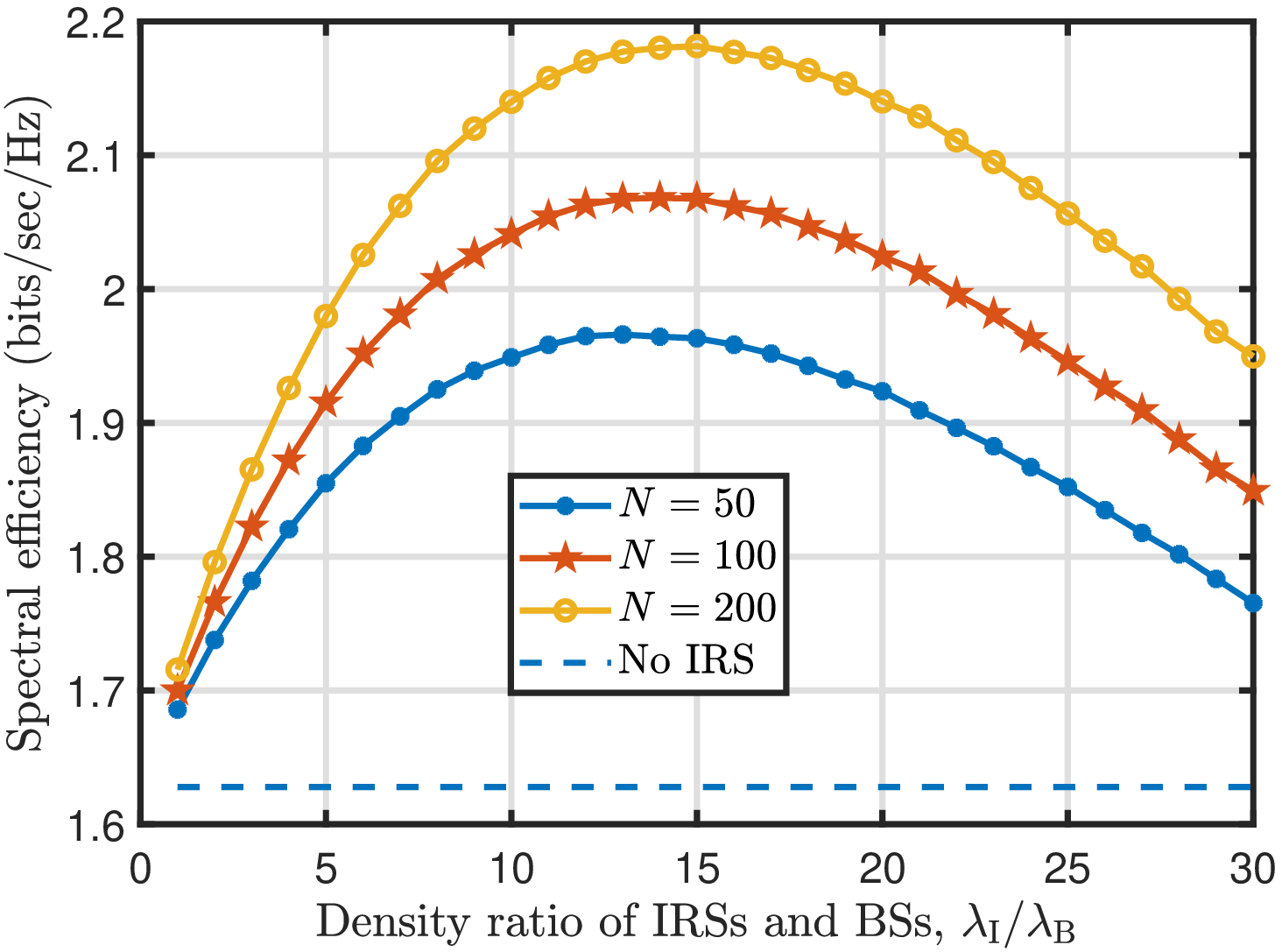}
			\end{minipage}
			\label{fig:spectralE_densityR}
		}~
		\subfigure[Tradeoff between the distance of IRS-UE and the number of IRS elements.]{
			\begin{minipage}[t]{0.49\textwidth}
				\includegraphics[width=\linewidth]{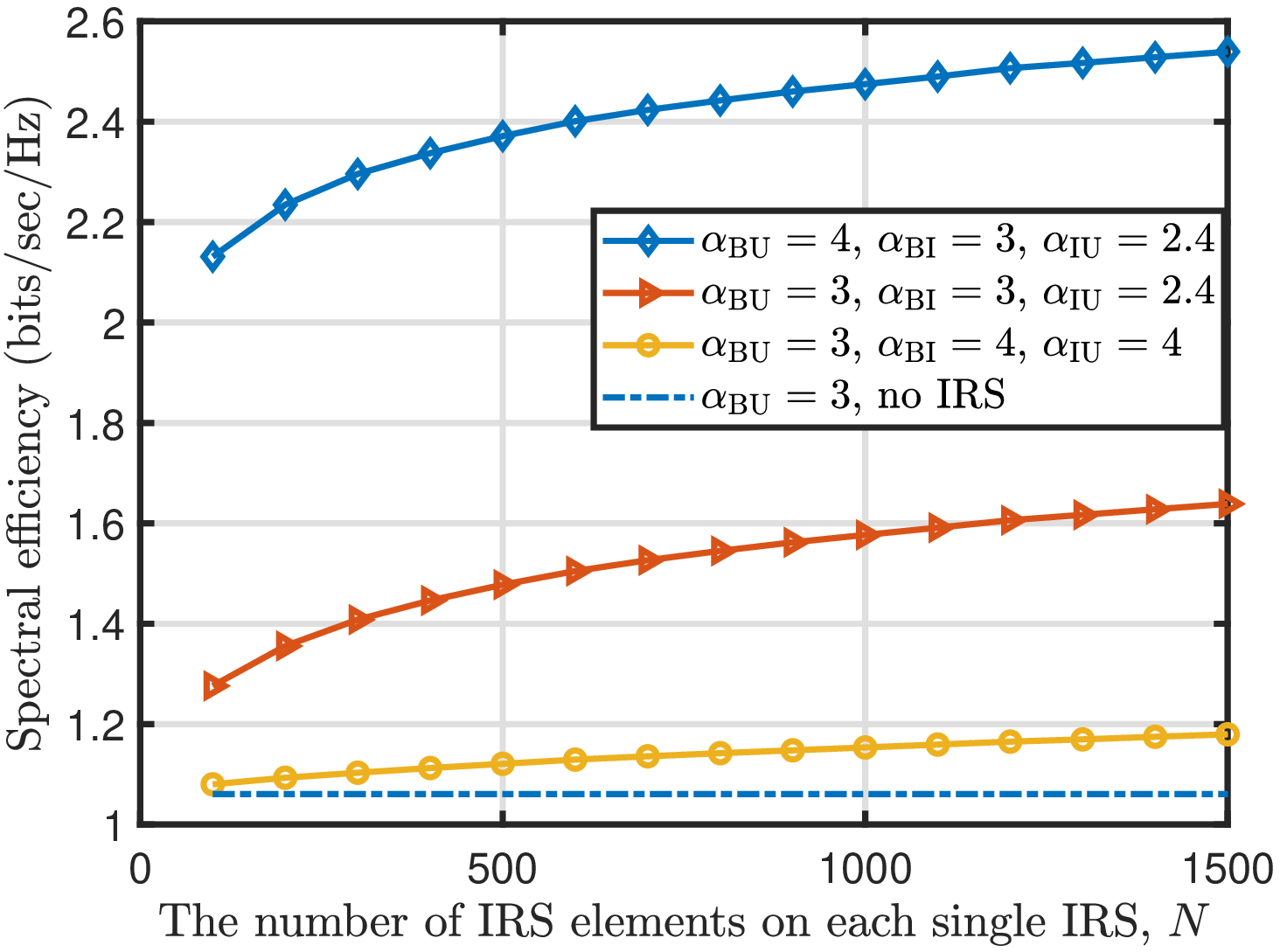}
			\end{minipage}
			\label{fig:spectralE_densityAlpha}
		}
    	\subfigure[ {Outage probability with different path-loss exponents, $\alpha$}.]{
    		\begin{minipage}[t]{0.49\textwidth}
    			\includegraphics[width=\linewidth]{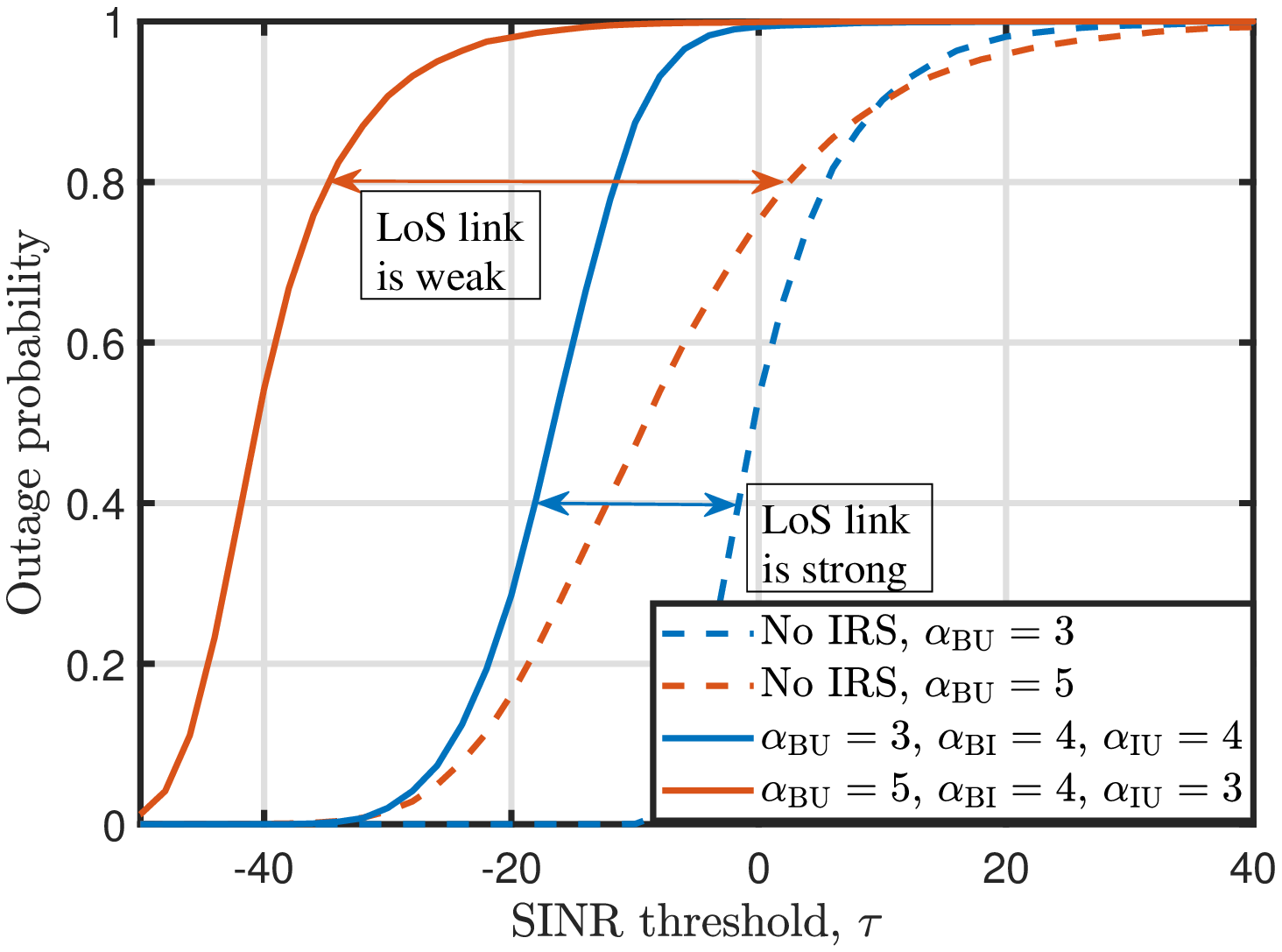}
    		\end{minipage}
    		\label{fig:outageP_diffAlpha}
    	}~
    	\subfigure[ {Outage probability with different IRS density, $\lambda_{\rm I}$, and number of IRS elements of each IRS, $N$}.]{
    		\begin{minipage}[t]{0.49\textwidth}
    			\includegraphics[width=\linewidth]{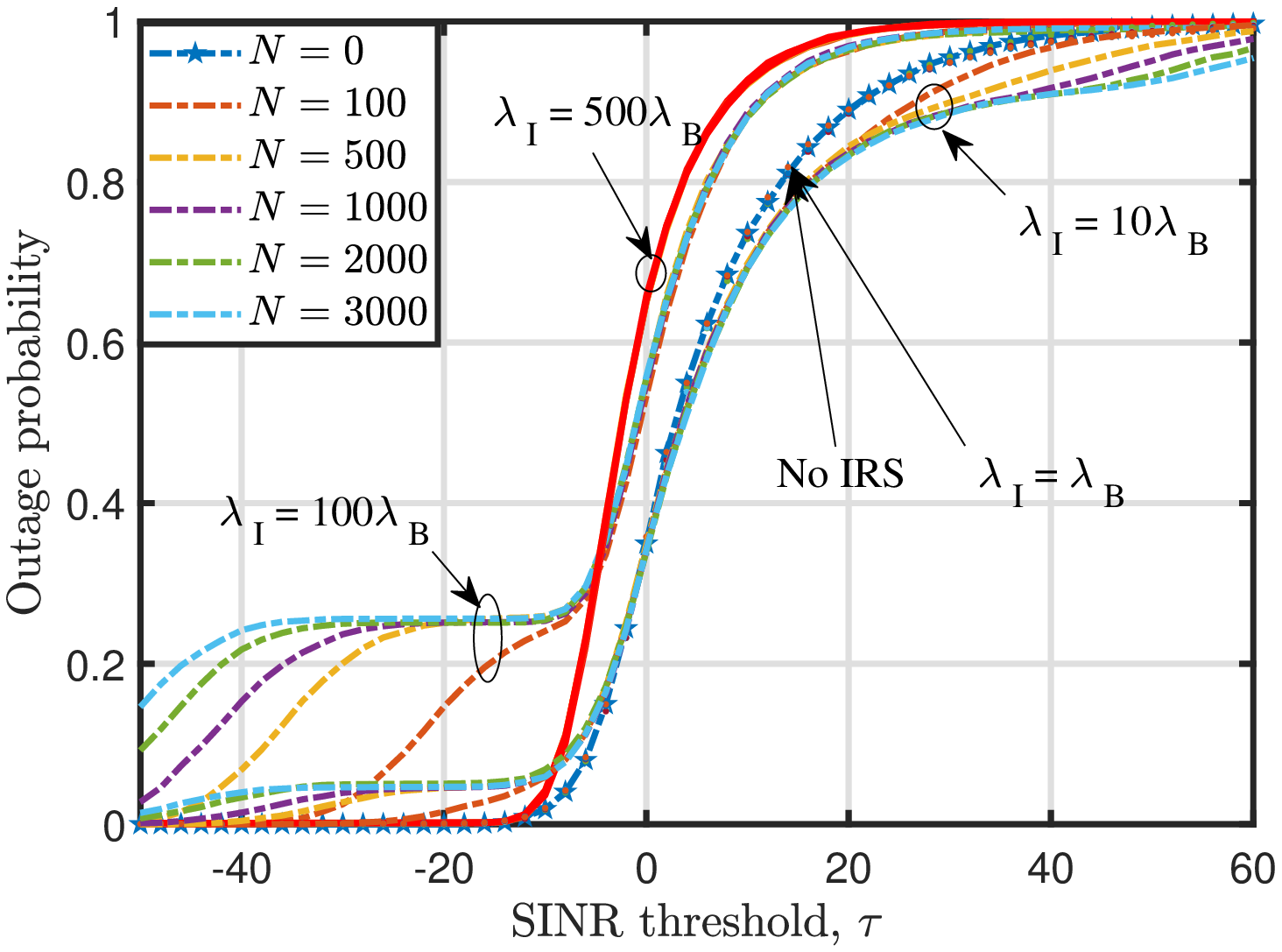}
    		\end{minipage}
    		\label{fig:outageP_diffIrsDensity}
    	}
		\caption{System performance analysis.}
		\label{fig:performance}
	\end{figure*}

	In this section, we introduced numerical results to verify the theoretical analysis. First, we displayed approximation results of cascaded channel gain and mixture channel gain compared with numerical results. Next, we illustrated CDF of the received signal distribution given link distance. Furthermore, the tradeoff between link distance and the number of IRS elements is provided. Finally, we showed the performance metrics. All of the simulations were carried out using MATLAB with the following parameters: BS density $\lambda_{\rm B}=1\times 10^{-5} {\rm /m^2}$, IRS density $\lambda_{\rm I}=1\times 10^{-4}{\rm /m^2}$, $D_1=25 \,{\rm m}$, $D_2=50\, {\rm m}$, unit transmit power $P_{\rm t} = 1\, {\rm Watt}$, noise power $\delta^{2}=-147$ dBm, and the number of elements of each IRS $N=500$. To better understand the mixture propagation environments and the system-level performance, the numerical analysis is carried out with Nakagami-$m$ fading channels as the underlying fading channel, if not specified otherwise.

	\subsection{Channel modeling of the channel gains}
	
	In Fig.~\ref{fig:MG_double}, we compared the mixture Gamma approximated PDF and CLT approximated PDF with the numerical PDF of the cascaded channel, from which we can observe that the PDF of mixture Gamma distribution fits well on cascaded fading channel gains compared with the CLT approximation, which validated the derivation in Theorem~\ref{theorem:Multipliability}. In addition, we displayed the minimum mean-square error (MMSE) of the mixture Gamma approximation of the cascaded channel gain with multiple fading types. In Fig.~\ref{fig:MG_mmse}, we can see that when the number of Gamma components of the mixture Gamma distribution is larger than 20, the accuracy can achieve $10^{-5}$ for most fading types. Moreover, in Fig.~\ref{fig:MG_mixture}, we compared the numerical PDF, mixture Gamma approximated PDF, and CLT approximated PDF of the mixture channel, which verified that the derivation of the PDF of the mixture channel in Theorem~\ref{theorem:Mixture}. At last, in Fig.~\ref{fig:MG_all_link}, the PDF of all types of channel gain are displayed, i.e., single channel gain, cascaded channel gain and mixture channel gain. We can observe that the cascaded channel gain is more concentrated than the single link while the mixture channel gain is more flat.

	\subsection{The received signal power distribution given link distances}
	
	In Fig.~\ref{fig:sigP_d_IU}, we displayed the CDF of the received signal power given the link distances. By varying the link distance between the typical UE and its serving IRS, $d_{\rm IU}$, we can observe that the benefit of decreasing $d_{\rm IU}$ is more significant when $d_{\rm IU}$ is smaller due to severe productive path loss, which is coincidence with \cite{lyu2021hybrid}. 
	Besides, in Fig.~\ref{fig:sigP_tradeoffND}, we illustrated the tradeoff between $d_{\rm IU}$ and the number of IRS elements of each IRS, $N$. The results are expected since the passive IRS suffers severe productive path loss, which hugely degrades the performance gain. We can see that about 8-fold the number of IRS elements is required to compensate for the path loss caused by large link distance.

	\subsection{System performance analysis}

	In Fig.~\ref{fig:spectralE_densityR}, the spectral efficiency versus the density ratio of IRSs and BSs, $ \lambda_{\rm I}/\lambda_{\rm B} $, with different numbers of IRS elements, $ N $, are displayed. We observe that the IRSs always provide spectral efficiency boosting. Besides, we can observe an optimal $ \lambda_{\rm I}/\lambda_{\rm B} $ under each $ N $, which slightly increases with the density ratio. This is expected since the IRSs enhance both signals and interference but with different scaling order \cite{lyu2021hybrid}. In Fig.~\ref{fig:spectralE_densityAlpha}, the effect of the path loss exponents of each link is displayed. We can observe that when the path loss of the BS$\rightarrow$IRS link and IRS$\rightarrow$UE link is severe, the performance improved by IRS is limited due to the productive path loss of the cascaded path. Besides, an interesting result is that worse BS$\rightarrow$UE link leads to a better performance, which is understandable since the interference is largely eliminated.

	Next, we displayed the outage probability in Fig.~\ref{fig:outageP_diffAlpha} and Fig.~\ref{fig:outageP_diffIrsDensity}. In Fig.~\ref{fig:outageP_diffAlpha}, we can observe that IRS boots the outage probability especially when the direct link is weak, which is the same as spectral efficiency, as shown in Fig.~\ref{fig:spectralE_densityAlpha}. In Fig.~\ref{fig:outageP_diffIrsDensity}, we studied the impact of density ratios of IRSs and BSs on outage probability. Surprisingly, the results differ from spectral efficiency, where the IRS definitely boosts the spectral efficiency. Several interesting observations are made as follows. First, when IRS density is small ($\lambda_{\rm I} = \lambda_{\rm B}$), the outage probability enhancement of IRS disappears. Second, when slightly increasing the IRS density, ($ \lambda_{\rm I}=10\lambda_{\rm B} $), although the outage probability at low SINR region is improved, within the high SINR region, the outage probability is decreased. The trend is that at low SINR region, larger number of IRS elements improves more significant enhancement of outage probability. In contrast, at the high SINR region, larger number of IRS elements decrease more significant enhancement of outage probability. This is reasonable since when amplifying the received signal, the interference is also enlarged. Third, when the IRS density is large enough ($\lambda_{\rm I} = 500 \lambda_{\rm B}$), the outage probability does not relied on the number of IRS elements on each IRS since the product path loss decreased due to small node distances. Fourth, if you want to increase the outage probability at low SINR, you can achieve this by increasing the IRS density and the number of IRS elements on each IRS smartly, such as, $\lambda_{\rm I} = 100\lambda_{\rm B}$ and $N=3000$, which can contribute to improving the reliability. However, at high SINR region, there is an optimal density ratio for improving the outage probability.

	\section{CONCLUSIONS AND DISCUSSIONS} \label{sectionConclusion}
	
	We proposed a uniform mixture Gamma channel modeling method for analyzing IRS-assisted wireless communication systems under any arbitrary fading environment and arbitrary number of links in this work. First, we proposed a method to approximate the channel gain of both the direct link and cascaded link independently over the mixture Gamma distributions for any arbitrary fading types by proving the multipliability of mixture Gamma distributed channels. Furthermore, the mixture channel is also modeled as a mixture Gamma distribution by proving the quadratic form of mixture Gamma distributed channels. Additionally, the fractional moments, Laplace transform, and CDF of mixture Gamma distribution are provided to facilitate the analysis. Second, we derived the mean conditional received signal power distribution and the Laplace transform of aggregated interference, under three operation modes: (a) one IRS associated with typical UE and other related IRSs scattering, (b) all related IRS randomly scattering, (c) no related IRS, given the distance between ${\rm BS}_{0}$ and ${\rm UE}_{0}$ and the distance between ${\rm IRS}_{0}$ and ${\rm UE}_{0}$. Finally, we introduced a uniform stochastic geometric system performance analysis framework based on the mixture Gamma distributed channels and derived the performance metrics availing of their corresponding SINR functions. In this way, the spectral efficiency, SINR moments, and outage probability are evaluated. Besides, the Monte-Carlo simulation verified the analysis and useful insights on system design are provided. We note that the properties of mixture Gamma distribution enormously facilitate the modeling and analysis in IRS-assisted networks with high accuracy. This modeling method also provides new insights on $K$ cascaded channels and $K$ mixture channels, which can be applied to multiple types of networks, such as active IRS networks, UAV networks, and relay networks.

	\appendices

	\section{} \label{appendix:CascadedMG}
    
	In this appendix, we provide a proof for Theorem~\ref{theorem:Multipliability}.
	In \cite{bhargav2018product}, the authors show that the PDF of a product of two random variables whose PDFs are linear combinations of Gamma distributions could be expressed by the Meijer-G function. Thus, the step (a) in \eqref{eq:fy} is achieved. However, the Meijer-G function is lack of tractability. As such,  by employing $G_{0,2}^{2,0}\left( y| b,c\right) = 2y^{\frac{1}{2}(b+c)}K_{b-c}(2\sqrt{y})$, we first simplified this to modified Bessel function in step (b). 
		\begin{equation}
		    \begin{split}
			f_{\rm Y}(y) \overset{(a)}{=} &~  \sum_{m_{1}=0}^{M_{1}} \sum_{{m_{2}}=0}^{M_{2}} \omega_{m_{1}} \omega_{m_{2}} \frac{\beta_{m_{1}} \beta_{m_{2}}}{\Gamma(\upsilon_{m_{1}} )\Gamma(\upsilon_{m_{2}})}G_{0,2}^{2,0}\left( y\beta_{m_{1}} \beta_{m_{2}} | \upsilon_{m_{1}} -1, \upsilon_{m_{2}}-1 \right)  
			\\ \overset{(b)}{=} &~  \sum_{{m_{1}} =0}^{M_{1}} \sum_{{m_{2}}=0}^{M_{2}}  \omega_{m_{1}}  \omega_{m_{2}} \frac{\beta_{m_{1}} \beta_{m_{2}} }{\Gamma(\upsilon_{m_{1}} )\Gamma(\upsilon_{m_{2}})}\times2(y\beta_{m_{1}} \beta_{m_{2}})^{\frac{1}{2}(\upsilon_{m_{1}} +\upsilon_{m_{2}})}  \times K_{\upsilon_{m_{1}} -\upsilon_{m_{2}}}(2\sqrt{y\beta_{m_{1}} \beta_{m_{2}}})
		    \end{split}\label{eq:fy},
		\end{equation}   
		Furthermore, the modified Bessel function $K_v(y)$ can be expressed as follows
		\begin{equation}
			K_{\upsilon_{m_{1}}-\upsilon_{m_{2}}}(2\sqrt{y\beta_{m_{1}}\beta_{m_{2}}}) =  \frac{1}{2}\left({y\beta_{m_{1}}\beta_{m_{2}}}\right)^{\frac{\upsilon_{m_{1}}-\upsilon_{m_{2}}}{2}}  \cdot \begin{matrix} \\ \underbrace{ \int_0^\infty \exp \left(-t-\frac{y\beta_{m_{1}}\beta_{m_{2}}}{t}\right){t^{-\upsilon_{m_{1}}+\upsilon_{m_{2}}-1}}\mathrm {d}t} \\I(t)\end{matrix} \label{appendix_eq:ModifiedBessel},
		\end{equation}
	where $I(t)$ can be derived as
	\begin{equation}
		I(t) = \int_0^\infty e^{-t}g(t)\mathrm {d}t, 
	\end{equation}        
	and $g(t)$ is given by:
	\begin{equation}
		g(t) = {\exp\left(-\frac{y\beta_{m_{1}}\beta_{m_{2}}}{t}\right){t^{-\upsilon_{m_{1}}+\upsilon_{m_{2}}-1}}}. 
	\end{equation}
	To solve this integration tractably, the modified Bessel function is approximated by Gaussian-Hermite functions with $\int_0^\infty e^{-t}g(t)dt \approx \sum_{i=0}^{I} \varpi_{i} g(t_i)$. As such, the PDF of  {$Y$} is achieved in \eqref{eq:fy1} with the restriction that absolute phase value of  {$y$} is no large than $\frac{1}{4}\pi$. 
	\begin{equation}
	    \begin{split}
	        f_{\rm Y}(y) =&~  \sum_{{m_{1}} =0}^{M_{1}} \sum_{{m_{2}}=0}^{M_{2}}  \frac{ \omega_{m_{1}}  \omega_{m_{2}}\beta_{m_{1}} \beta_{m_{2}} }{\Gamma(\upsilon_{m_{1}} )\Gamma(\upsilon_{m_{2}})}(y\beta_{m_{1}} \beta_{m_{2}})^{\upsilon_{m_{1}} }   \int_0^\infty {\exp\left(-t-\frac{y\beta_{m_{1}} \beta_{m_{2}}}{t}\right)\frac{1}{t^{\upsilon_{m_{1}} -\upsilon_{m_{2}}+1}}} \mathrm {d}t  \\ = &~ \sum_{i=1}^{I}  \sum_{{m_{1}} =0}^{M_{1}} \sum_{{m_{2}}=0}^{M_{2}} \omega_{m_{1}}  \omega_{m_{2}} \frac{(\beta_{m_{1}} \beta_{m_{2}})^{\upsilon_{m_{1}} +1} }{\Gamma(\upsilon_{m_{1}} )\Gamma(\upsilon_{m_{2}})}{t_{i}^{-\upsilon_{m_{1}} +\upsilon_{m_{2}}-1}} \varpi_{i} y^{\upsilon_{m_{1}} }e^{-\frac{\beta_{m_{1}} \beta_{m_{2}}}{t_{i}}y}
	    \end{split}, \label{eq:fy1} 
	\end{equation}
	With some mathematical simplifications, the double-link distribution PDF could be simplified as a mixture Gamma distribution as shown in \eqref{MGdouble}.  This completes the proof.
	

	\section{} \label{appendix:TheoremMixtureChannel}

	In this appendix, we provide a proof for Theorem~\ref{theorem:Mixture}. If $X^{2}$ follows mixture Gamma distribution, the PDF of $X$ is given by
    \begin{equation}\label{eq:X_naka}
        f_{\rm X}(x)=2\sum_{i=1}^{I}\varepsilon_{i}x^{2\beta_{i}-1}e^{-\xi_{i}x^{2}}.
    \end{equation}
    The PDF of $Z=X+Y$ can be derived using convolution and \eqref{eq:X_naka}, as shown in \eqref{eq:mixture1}. With some simple mathematical simplifications, (a) is achieved by invoking  $(a+x)^{n}=\sum_{k=0}^{n}\tbinom{n}{k}x^{k}a^{n-k}$
	\begin{small}
	\begin{equation}
	\begin{split}
	    f_{\rm Z}(z) = &~ \int_{0}^{z}{f_{\rm X}(x)f_{\rm Y}(z-x)}\mathrm{d}x  = \int_{0}^{z}{ \left (2\sum_{i=1}^{I}\varepsilon_{i}x^{2\beta_{i}-1}e^{-\xi_{i}x^{2}}\right)\left(2\sum_{j=1}^{J}\varepsilon_{j}(z-x)^{2\beta_{j}-1}e^{-\xi_{j}(z-x)^{2}}\right)} \mathrm{d}x \\ \overset{(a)}{=} &~ \sum_{i=1}^{I} \sum_{j=1}^{J} 4\varepsilon_{i} \varepsilon_{j}  \sum_{k_{1}=0}^{2\beta_{j}-1} \tbinom{2\beta_{j}-1}{k_{1}} (-1)^{k_{1}} z^{2\beta_{j}-k_{1}-1} e^{ z^{2}\left( \frac{\xi_{j}^{2}}{\xi_{i}+\xi_{j}} - \xi_{j}\right) }  \begin{matrix} \\ \underbrace{\int_{0}^{z} {x^{k_{1}+2\beta_{i}-1} e^{-(\xi_{i}+\xi_{j})\left(x-\frac{\xi_{j}}{\xi_{i}+\xi_{j}}z\right)^{2}}} \mathrm{d}x, } \\g^{'}(x)\end{matrix} 
	\end{split}\label{eq:mixture1}
	\end{equation}
	\end{small}%
	where $g^{'}(x)$ can be further derived by substituting $t = \left(x-\frac{\xi_{j}^{2}}{\xi_{i}+\xi_{j}}\right)^2$ and  \cite[eq.(3.381.2)]{zwillinger2007table}
	\begin{small}
	\begin{equation}
			\begin{split}
				g^{'}(x) = &~ \int_{0}^{\left(\frac{\xi_{j}}{\xi_{i}+\xi_{j}}z\right)^2} {\left(\frac{\xi_{j}}{\xi_{i}+\xi_{j}}z-\sqrt{t}\right)^{k_{1}+2\beta_{i}-1} }\frac{e^{-(\xi_{i}+\xi_{j})t}}{2\sqrt{t}} + {\left(\frac{\xi_{j}}{\xi_{i}+\xi_{j}}z+\sqrt{t}\right)^{k_{1}+2\beta_{i}-1} }\frac{e^{-(\xi_{i}+\xi_{j})t}}{2\sqrt{t}} \mathrm{d}t \\ = &~ \frac{1}{2}\sum_{k_{2}=0}^{2\beta_{i}-1+k_{1}}\tbinom{2\beta_{i}-1+k_{1}}{k_{2}}\frac{(\xi_{j}z)^{2\beta_{i}-1+k_{1}-k_{2}}}{(\xi_{i}+\xi_{j})^{2\beta_{i}+k_{1}-\frac{k_{2}+1}{2}}}\left[ \gamma\left( \frac{k_{2}+1}{2},\frac{\xi_{j}^{2}z^2}{\xi_{i}+\xi_{j}} \right) + \gamma\left( \frac{k_{2}+1}{2},\frac{\xi_{i}^{2}z^2}{\xi_{i}+\xi_{j}} \right)\right].
			\end{split}\label{eq:mixture_gx}
		\end{equation}	
	\end{small}%

	    Then, by substituting \eqref{eq:mixture_gx} into \eqref{eq:mixture1}, the PDF of $Z$ is derived as follows
		\begin{equation}
		\begin{split}
			 f_{\rm Z}(z) = &~   \sum_{i=1}^{I} \sum_{j=1}^{J} 2\varepsilon_{i} \varepsilon_{j}e^{ - \frac{\xi_{i}\xi_{j}}{\xi_{i}+\xi_{j}}z^{2} }  \sum_{k_{1}=0}^{2\beta_{j}-1} \tbinom{2\beta_{j}-1}{k_{1}} (-1)^{k_{1}}  \sum_{k_{2}=0}^{2\beta_{i}+k_{1}-1}\tbinom{2\beta_{i}+k_{1}-1}{k_{2}} z^{2\beta_{i}+2\beta_{j}-2-k_{2}} \\ &~ \cdot \frac{\xi_{j}^{2\beta_{i}-1+k_{1}-k_{2}}}{(\xi_{i}+\xi_{j})^{2\beta_{i}+k_{1}-\frac{k_{2}+1}{2}}} \left[(-1)^{k_{2}} \gamma\left(\frac{k_{2}+1}{2},\frac{\xi_{j}^{2}z^{2}}{\xi_{i}+\xi_{j}} \right) + \gamma\left(\frac{k_{2}+1}{2},\frac{\xi_{i}^{2}z^{2}}{\xi_{i}+\xi_{j}} \right) \right] \\ \overset{(a)}{=} &~ \sum_{i=1}^{I} \sum_{j=1}^{J} 2\varepsilon_{i} \varepsilon_{j} \sum_{k_{1}=0}^{2\beta_{j}-1} \tbinom{2\beta_{j}-1}{k_{1}} (-1)^{k_{1}}  \sum_{k_{2}=0}^{2\beta_{i}+k_{1}-1}\tbinom{2\beta_{i}+k_{1}-1}{k_{2}} \frac{\Gamma(\frac{k_{2}+1}{2}){\xi_{j}^{2\beta_{i}-1+k_{1}-k_{2}}}z^{2\beta_{i}+2\beta_{j}+2k_{3}-2}}{\Gamma(\frac{k_{2}+1}{2}+k_{3}+1){(\xi_{i}+\xi_{j})^{2\beta_{i}+k_{1}+k_{3}}}}  \\ &~ \cdot \left[ (-1)^{k_{2}}\xi_{j}^{k_{2}+1+2k_{3}}e^{-\xi_{j}z^{2}} +\xi_{i}^{k_{2}+1+2k_{3}}e^{-\xi_{i}z^{2}} \right] ,
		\end{split}\label{eq:MGcomposite}
		\end{equation}
        where we applied the power series expansion of incomplete Gamma function in step (a). Last, we apply $f_{\rm S}(s) = \frac{1}{2}s^{-\frac{1}{2}}f_{\rm Z}(\sqrt{s})$ to derive the PDF of $S$. This completes the proof.

    \section{} \label{appendix:LaplaceI}

        This appendix provides derivation of $ \mathcal{L}_{\rm I_{{\rm F},1}}| _{d_{\rm BIU}^{(0)}} $ and $ \mathcal{L}_{\rm I_{{\rm F},2}}| _{d_{\rm BIU}^{(0)}} $. Similar to \cite{chun2017comprehensive}, the Laplace transform of $I_{{\rm F},1}$ is derived in \eqref{eq:L_I1_appen}. In step (a), a substitution $t=sH_{\rm BU}{d_{\rm BU}}^{-\alpha}$ is applied. Next, in step (b) the integration is achieved through integration by parts. Then, the last step is achieved with some mathematical simplifications.
        \begin{small}
        \begin{equation}\label{eq:L_I1_appen}
        	\begin{split} 
        		\mathcal{L}_ {\rm I_{{\rm F},1}} |_{d_{\rm BU}^{(0)}} =&~ \exp\left( -2\pi\lambda_{\rm B}\int_{{d_{\rm BU}^{(0)}}}^{\infty} \left( 1-\mathbb{E}_{\rm H} \left[ e^{-sH_{\rm BU}d_{\rm BU}^{-\alpha}} \right]  \right)d_{\rm BU}{\rm d}d_{\rm BU}  \right) \\ \overset{(a)}{=} &~ \exp\left( -2\pi\lambda_{\rm B}  \frac{(sH_{\rm BU})^{\frac{2}{\alpha}}}{\alpha}\mathbb{E}_{\rm H} \left[\int_{0}^{sH_{\rm BU}{d_{\rm BU}^{(0)}}^{-\alpha}} \left( 1-e^{-t} \right) t^{-1-\frac{2}{\alpha}} {\rm d} t   \right] \right) \\ \overset{(b)}{=} &~ \exp\left( -2\pi\lambda_{\rm B}  \frac{(sH_{\rm BU})^{\frac{2}{\alpha}}}{2}\mathbb{E}_{\rm H} \left[ -t^{-\frac{2}{\alpha} }\left(1-e^{-t}\right)\Big|_{0}^{sH_{\rm BU}} +  \int_{0}^{sH_{\rm BU}{d_{\rm BU}^{(0)}}^{-\alpha}} t^{-\frac{2}{\alpha}} e^{-t}  {\rm d} t  \right] \right)   \\ = &~ \exp \left(  -\pi\lambda_{\rm B}{d_{\rm BU}^{(0)}}^{2}  \sum_{k=0}^{\infty} (-1)^{k}\frac{m_{\rm BU}^{-1-k}\left(s{d_{\rm BU}^{(0)}}^{-\alpha}\right)^{1+k}}{\Gamma(m_{\rm BU}  )k!(1-\frac{2}{\alpha}+k)} \Gamma(1+k+m_{\rm BU}  )   \right)  \\  &~ \cdot \exp \left(  -\pi\lambda_{\rm B}{d_{\rm BU}^{(0)}}^{2} (s\eta)^{\frac{2}{\alpha}}\frac{m_{\rm BU}  ^{m_{\rm BU}  }}{\Gamma(m_{\rm BU}  )} \Gamma\left(\frac{2}{\alpha}+m_{\rm BU}  \right)\left( m_{\rm BU}  ^{\frac{2}{\alpha}+m_{\rm BU}  } - (s\eta {d_{\rm BU}^{(0)}}^{-\alpha}+m_{\rm BU}   )^{\frac{2}{\alpha}+m_{\rm BU}  } \right)  \right).
        	\end{split} 
        \end{equation}        
        \end{small}%
        Following similar procedure, $ \mathcal{L}_{\rm I_{{\rm F},2}}| _{d_{\rm BIU}^{(0)}} $ can be derived as \eqref{eq:L_I2_appen}. 
\begin{small}
	\begin{equation}\label{eq:L_I2_appen}
		\begin{split}
			\mathcal{L}_{\rm I_{{\rm F},2}}| _{d_{\rm BIU}^{(0)}}=&~ \mathbb{E}_{\Lambda _{B}\setminus\{0\}}\left\{ e^{-s \eta H_{\rm BIU}} \right\}\Big|_{d_{\rm BU}^{(0)}} \\ = &~ {\rm exp}\left(-2\pi \lambda_{\rm B}  \int_{0}^{D_{2}}\int_{d_{\rm BU}^{(0)}}^{\infty}\left({1-\mathbb{E}_{\rm H}\left[ e^{-s \eta H_{\rm BIU} } \right]}\right)d_{\rm BU} \,\,\mathrm {d} d_{\rm BU} \,\, d_{\rm IU} f_{\rm d_{\rm IU}}( d_{\rm IU} )\,\,\mathrm {d} d_{\rm IU}  \right)  \\ = &~ \exp \Bigg(  -\pi\lambda_{\rm B} \int_{0}^{D_{2}} \mathbb{E}_{\rm H} \Big[ \left(s\eta H_{\rm BIU}\right)^{\frac{2}{\alpha}} \gamma\left(1-\frac{2}{\alpha},s\eta H_{\rm BIU}({d_{\rm BU}^{(0)}}d_{\rm IU})^{-\alpha}\right)  \\&~  - \left({d_{\rm BU}^{(0)}} d_{\rm IU}\right)^{2} \left(1-e^{-s\eta H_{\rm BIU} ({d_{\rm BU}^{(0)}} d_{\rm IU})^{-\alpha} }\right)  \Big]   f_{\rm d_{\rm IU}}( d_{\rm IU} )\,\,\mathrm {d} d_{\rm IU}  \Bigg).
		\end{split}
	\end{equation} 
\end{small}%

    \section{} \label{appendix:Perf}
        
        In this appendix, we provide a proof for Theorem~\ref{theoremPerfAnalysis}. In \eqref{ESINR}, $\mathbb{E}[g({\rm SINR})]$ is derived. The step (a) is achieved by substitute $z$ with $z=\frac{H_{\rm S}}{I_{\rm F}+\delta^{2}}$. Then the step (b) is achieved by substitute $b$ with $b=\xi_{i} {({I_{\rm F}+\delta^{2}})}$ 
     \begin{small}
        \begin{equation}
        \begin{split}
        	\mathbb{E}[g( {\rm SINR})] = &~ \int_{0}^{\infty} {g( {\rm SINR})f_{\rm H_{\rm S}}(x)}\mathrm{d}x = \int_{0}^{\infty} {g\left( \frac{H_{\rm S}}{I_{\rm F}+\delta^{2}} \right) \sum_{i=0}^{N}}{\varepsilon_{i}H_{\rm S}e^{-\xi_{i}H_{\rm S}}}\mathrm{d}H_{\rm S}  \\ \overset{(a)}{=} &~ \sum_{i=0}^{N}{\varepsilon_{i} {({I_{\rm F}+\delta^{2}}) }^{\beta_{i}} \int_{0}^{\infty} {g(z)z^{\beta_{i}-1} e^{-\xi_{i} z {({I_{\rm F}+\delta^{2}})}}}\mathrm{d}z }   \\ \overset{(b)}{=} &~ \sum_{i=0}^{N} \varepsilon_{i} \Gamma(\beta_{i}) {\xi_{i}}^{-\beta_{i}} \begin{matrix} \\ \underbrace{  \int_{0}^{\infty} {g(z) \frac{z^{\beta_{i}-1}}{\Gamma(\beta_{i})}  b^{\beta_{i}} e^{-b {\xi_{i}} z}}\mathrm{d}z }  \\Q\end{matrix} \\ = &~ \sum_{i=0}^{N} \varepsilon_{i} \Gamma(\beta_{i}) {\xi_{i}}^{-\beta_{i}}  \int_{0}^{\infty} {g_{\beta_{i}}(z) e^{-\delta^{2} {\xi_{i}} z}} \mathcal{L}_{\rm I_{\rm F}}( {\xi_{i}}z)  \mathrm{d}z.  
        	\end{split}\label{ESINR} 
            \end{equation}          
     \end{small}%
     Next, $Q$ is evaluated in \eqref{Q} by utilizing the partial integral as follows 
        \begin{equation}
            Q = -\sum_{k=0}^{\beta_{i}-1}g_{k}(z)b^{\beta_{i}-k-1}e^{-bz}|_{0}^{\infty} + \int_{0}^{\infty} {g_{\beta_{i}}(z) e^{-b  z}}\mathrm{d}z. \label{Q} 
        \end{equation}
    This completes the proof.
 \section{} \label{appendix:Distance}
 
 In this appendix, we evaluated the PDF and CDF of the distance between IRS and BS. 
 \begin{figure}[t!]
 	\centering
 	\includegraphics[height=6.22cm,width=8cm]{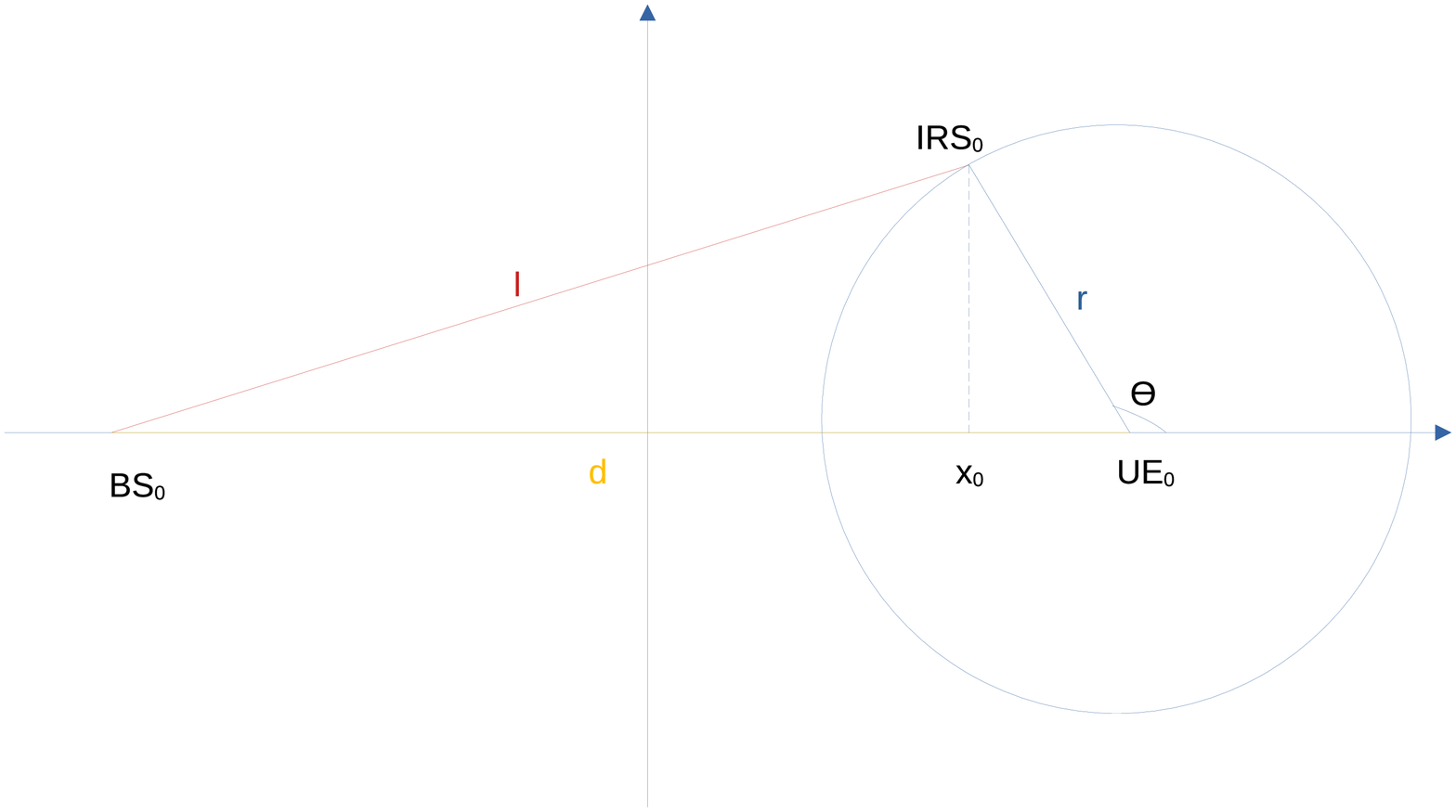}
 	\caption{Distance Model.}
 	\label{fig:model}
 \end{figure}
 As illustrated in Fig.~\ref{fig:model}, for given $r$ and $d$, ${\rm IRS}_{0}$ is located on a circle with $\rm UE_0$ as the center and $r$ as the radius. According to their geometric relationship, we can get the following equation
 
 \begin{equation}
 	l^2 = r^2 + d^2 + 2dr \cos\theta . 
 \end{equation}
 Assuming that $\theta$ is uniformly distributed over $[0, 2\pi]$, then the conditional PDF of $y = \cos \theta$ is given by
 
 \begin{equation}
 	f_{\rm Y}(y)|_{d,r} =  \left\{ \begin{array}{rcl}
 		\frac{1}{\pi \sqrt{1-y^2}},       &      & {y\in [-1,1]},\\
 		0,     & & \text{otherwise}.	\end{array} \right  .  
 \end{equation}
 Next,the conditional CDF of $l=\sqrt{2dry+r^2+d^2}$ is derived as 
 \begin{equation}
 	\begin{split}
 		F_{\rm L}(l)|_{d,r} = &~ P(L \leq l) = P(Y \leq \frac{l^2-r^2-d^2}{2dr}) \\ = &~ \int_{0}^{\frac{l^2-r^2-d^2}{2dr}} \frac{1}{\pi \sqrt{1-y^2}} \mathrm{d}y = \frac{1}{\pi} \arcsin\left( \frac{l^2-r^2-d^2}{2dr} \right) + \frac{1}{2}.
 	\end{split} 
 \end{equation}
 By taking the derivative of CDF, we can get its PDF
 
 \begin{equation}
 	\begin{split}
 		f_{\rm L}(l)|_{d,r} = \frac{\mathrm{d}}{\mathrm{d}l} F_{\rm L}(l) = \frac{l/dr}{\pi \sqrt{1-\left( \frac{l^2-r^2-d^2}{2dr} \right) ^2}},
 	\end{split} 
 \end{equation}
 and by substituting $x=\frac{l^2}{2dr}$ and $a=\frac{r^2+d^2}{2dr}$, its mean can be achieved as follows
 \begin{small}
 \begin{equation}\label{eq:mean} 
 \begin{split}
 	\mathbb{E}[l]|_{d,r} = &~ \frac{1}{\pi} \int_{d-r}^{d+r} \frac{l^2/dr}{\sqrt{1-\left( \frac{l^2-r^2-d^2}{2dr} \right) ^2}} \mathrm{d}l \overset{(a)}{=} \frac{2}{\pi} \int_{a-1}^{a+1}  \frac{x}{\sqrt{1-\left( x-a \right) ^2}} \mathrm{d}x  \\ = &~ \frac{2}{\pi} \Bigg[ \int_{a-1}^{a+1} \frac{x-a}{\sqrt{1-\left( x-a \right) ^2}} \mathrm{d}x + \int_{a-1}^{a+1} \frac{a}{\sqrt{1-\left( x-a \right) ^2}} \mathrm{d}x \Bigg]\notag \\ = &~ 2a = \frac{r^2+d^2}{dr}.
 \end{split} 
 \end{equation} 
 \end{small}%
 The integral of step (a) can be derived by two parts as shown in \eqref{eq:itegral1} and \eqref{eq:itegral2}, where the result in \eqref{eq:itegral1} is achieved by substituting $u=x-a$
 \begin{small}
  \begin{equation}
 	\begin{split}
 	 \int_{a-1}^{a+1} \frac{x-a}{\sqrt{1-\left( x-a \right) ^2}} \mathrm{d}x = \int_{-1}^{1} \frac{u}{\sqrt{1-u^2}} \mathrm{d}u = -\sqrt{1-u^2}|_{-1}^{1} = 0,
 	\end{split} \label{eq:itegral1}  
 \end{equation}
  \begin{equation}
 	\begin{split}
 		\int_{a-1}^{a+1} \frac{a}{\sqrt{1-\left( x-a \right) ^2}} \mathrm{d}x = a\cdot \arcsin (x-a)|_{a-1}^{a+1} = a\pi. 
 	\end{split} \label{eq:itegral2} 
 \end{equation}    
 \end{small}%
 The result $\mathbb{E}[l]|_{d,r}=\frac{r^2+d^2}{dr}$ does match with our assumption, $l \approx d$ when $d \gg r$, assuming that $r>1$. This completes the proof.

\section{} \label{appendix:IterationAlgorithmForChannel}

In this appendix, we provided the iteration algorithm for obtaining the PDF of the $K$ cascaded channels in Algorithm~\ref{alg:A},
\begin{algorithm}
	\caption{The PDF of $K$-cascaded channel gain}
	\label{alg:A}
	\begin{algorithmic}[1] 
		\Require The mixture Gamma distribution parameters of each link: $\Theta_{k} \colon = \{ \varepsilon_{m_{k}}, \beta_{m_{k}}, \xi_{m_{k}},M_{k} \}$, $M_{k}$ is the number of Gamma terms of the $k$-th link, $k=1,2,...,K$, $m_{k}=1,2,...,M_{k}$
		\Function{Cascade}{ {$\Theta_{k}$}
		}
		\For{$m_{1} = 1 \to M_{1}$}
		\For{$m_{2} = 1 \to M_{2}$}
		\For{$i = 1 \to I$}
		\State $\varepsilon_{m_{1},m_{2},i}^{(2)} = \frac{\varepsilon_{m_{1}}\varepsilon_{m_{2}}(\xi_{m_{1}}\xi_{m_{2}})^{\beta_{m_{1}}}}{\xi_{m_{1}}^{\beta_{m_{1}}}\xi_{m_{2}}^{\beta_{m_{2}}}}\varpi_{i}t_{i}^{-\beta_{m_{1}}+\beta_{m_{2}}-1}$, 
		\State $\beta_{m_{1},m_{2},i}^{(2)} = \varepsilon_{m_{1}}$, 
		\State $\xi_{m_{1},m_{2},i}^{(2)} =  \frac{\xi_{m_{1}}\xi_{m_{2}}}{t_{i}}$
		\EndFor
		\EndFor
		\EndFor
		\State $M^{(2)}=M_{1}M_{2}I$
		\For{$k=3\to K$}
		\For{$m_{1} = 1 \to M^{(k-1)}$}
		\For{$m_{2} = 1 \to N_{k}$}
		\For{$i = 1 \to I$}
		\State $\varepsilon_{m_{1},m_{2},i}^{(k)} = \frac{\varepsilon_{m_{1}}^{(k-1)}\varepsilon_{m_{k}}(\xi_{m_{1}}^{(k-1)}\xi_{m_{k}})^{\beta_{m_{1}}^{(k-1)}}}{{\xi_{m_{1}}^{(k-1)}}^{\beta_{m_{1}}^{(k-1)}}\xi_{m_{k}}^{\beta_{m_{k}}}}\varpi_{i}t_{i}^{-\beta_{m_{1}}+\beta_{m_{2}}-1}$, 
		\State $\beta_{m_{1},m_{2},i}^{(k)} = \varepsilon_{m_{1}}^{(k-1)}$, 
		\State $\xi_{m_{1},m_{2},i}^{(k)} = \frac{\xi_{m_{1}}^{(k-1)}\xi_{m_{2}}^{(k-1)}}{t_{i}}$
		\EndFor
		\EndFor
		\EndFor
		\State $M^{(k)}=M^{(k-1)}M_{k}I$
		\EndFor
		\State \Return{$\Theta^{(k)}\colon=\{\varepsilon^{(k)},\beta^{(k)},\xi^{(k)},M^{(k)} \}$}
	\EndFunction
\end{algorithmic}
\end{algorithm}
Moreover, the iteration algorithm can be straightforwardly obtained following a similar iteration procedure as Algorithm~\ref{alg:A}.

    \begin{figure*}
        \centering
        \includegraphics[angle=90,scale=0.55]{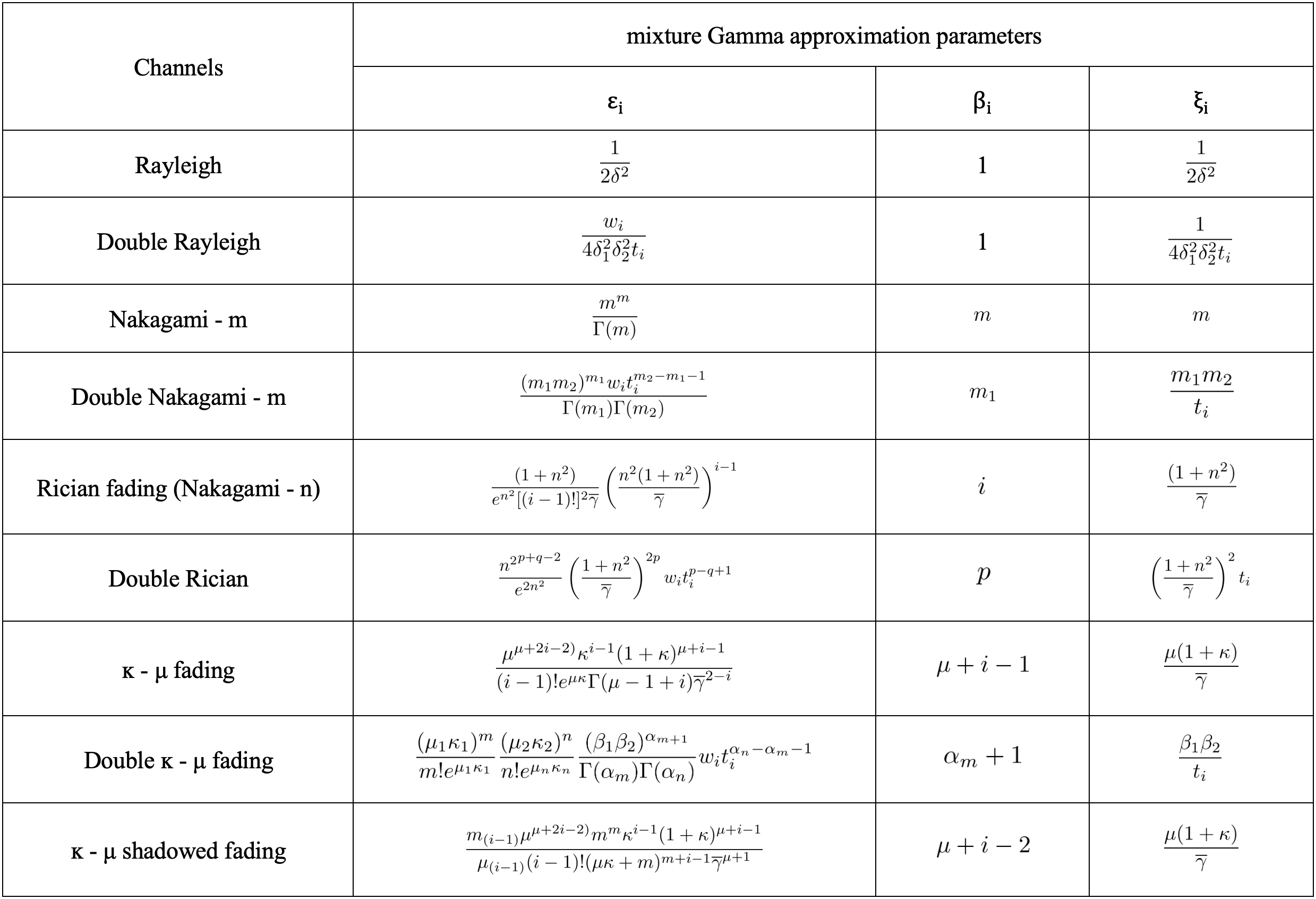}
        \caption{Multiple channels and their mixture Gamma approximation table}
        \label{fig:table}
    \end{figure*}
        
\bibliographystyle{IEEEtran}  
\bibliography{references}






\end{document}